\documentclass{article}

\PassOptionsToPackage{numbers}{natbib}
\usepackage[preprint]{neurips_2019}
\usepackage[utf8]{inputenc} 
\usepackage[T1]{fontenc}    
\usepackage{hyperref}       
\usepackage{url}            
\usepackage{booktabs}       
\usepackage{amsfonts}       
\usepackage{nicefrac}       
\usepackage{microtype}      

\renewcommand{\thefootnote}{\alph{footnote}}

\newcommand{\symfootnote}[1]{%
\let\oldthefootnote=\thefootnote%
\stepcounter{mpfootnote}%
\addtocounter{footnote}{-1}%
\renewcommand{\thefootnote}{\fnsymbol{mpfootnote}}%
\footnote{#1}%
\let\thefootnote=\oldthefootnote%
}

\title{Low-Variance and Zero-Variance Baselines for Extensive-Form Games}

\author{Trevor Davis,\textsuperscript{1\textdagger}
Martin Schmid,\textsuperscript{2}
Michael Bowling,\textsuperscript{2,1}\\
\textsuperscript{1}Department of Computing Science, University of Alberta\\
\textsuperscript{2}DeepMind\\
trdavis1@ualberta.ca,
\{mschmid, bowlingm\}@google.com\\
\textsuperscript{\textdagger} Work done during an internship at DeepMind.}

\usepackage{appendix}
\usepackage{comment}
\usepackage{times}
\usepackage{xcolor}
\usepackage{soul}
\usepackage[small,skip=5pt]{caption}
\usepackage{etoolbox}

\newtoggle{appendix}
\toggletrue{appendix}

\usepackage{algorithm,algpseudocode}
\usepackage{amsmath,amssymb,amsthm}
\usepackage{bm, bbm}
\usepackage{color}
\usepackage{mathtools}
\usepackage{multicol}
\usepackage{siunitx}
\usepackage[skip=5pt]{subcaption}
\usepackage{varwidth}
\usepackage[disable]{todonotes}

\newcommand{\Reals}{\mathbb{R}}
\newcommand{\Prob}[1]{\Pr\left[#1\right]}
\newcommand{\Exp}[2][]{\mathbb{E}_{#1}\left[#2\right]}
\newcommand{\Var}[2][]{\mathrm{Var}_{#1}\left[#2\right]}
\newcommand{\indic}{\mathbbm{1}}
\newcommand{\BigO}[1]{\mathcal{O}(#1)}

\newcommand{\Players}{N}
\newcommand{\player}{i}
\newcommand{\my}[1]{#1_i}
\newcommand{\opp}[1]{#1_{-i}}
\newcommand{\chance}{c}
\newcommand{\chances}[1]{#1_c}
\newcommand{\Histories}{H}
\newcommand{\hist}{h}
\newcommand{\Actions}{A}
\newcommand{\act}{a}
\newcommand{\Terminals}{Z}
\newcommand{\term}{z}
\newcommand{\Toact}{P}
\newcommand{\utility}{u}
\newcommand{\exputil}[2]{\utility(#1|#2)}

\newcommand{\Isets}{\mathcal{I}}
\newcommand{\Iset}{I}
\newcommand{\Strategies}{\Sigma}
\newcommand{\strategy}{\sigma}
\newcommand{\avgstrat}{\overline{\strategy}}
\newcommand{\reach}[1]{\pi^{#1}}
\renewcommand{\time}{t}
\newcommand{\Time}{T}
\newcommand{\Timesteps}{\mathcal{T}}
\newcommand{\regret}[1]{r^{#1}}
\newcommand{\sampreg}[3]{\hat{r}^{#1}(#2|#3)}
\newcommand{\Regret}[1]{R^{#1}}
\newcommand{\sampling}[1]{q^{#1}}
\newcommand{\samputil}[3]{\hat{\utility}(#1|#2,#3)}
\newcommand{\baseline}{b}
\newcommand{\baseutil}[3]{\hat{\utility}_{\baseline}(#1|#2,#3)}
\newcommand{\weight}[1]{w_{#1}}
\newcommand{\Pubstate}{S}
\newcommand{\Pubstates}{\mathcal{S}}
\newcommand{\Transitions}[1]{\mathcal{T}(#1)}

\newcommand{\timefunc}[1]{\text{time}(#1)}
\newcommand{\partutil}[2]{\tilde{\utility}^{#1}(#2)}

\newtheorem{thm}{Theorem}

\newtheorem{lem}{Lemma}

\algnewcommand{\Iif}[2]{\State \algorithmicif\ #1\ \algorithmicthen\ #2}

\begin{document}
\maketitle

\begin{abstract}
Extensive-form games (EFGs) are a common model of multi-agent interactions with imperfect information. State-of-the-art algorithms for solving these games typically perform full walks of the game tree that can prove prohibitively slow in large games. Alternatively, sampling-based methods such as Monte Carlo Counterfactual Regret Minimization walk one or more trajectories through the tree, touching only a fraction of the nodes on each iteration, at the expense of requiring more iterations to converge due to the variance of sampled values. In this paper, we extend recent work that uses baseline estimates to reduce this variance. We introduce a framework of baseline-corrected values in EFGs that generalizes the previous work. Within our framework, we propose new baseline functions that result in significantly reduced variance compared to existing techniques. We show that one particular choice of such a function --- predictive baseline --- is provably optimal under certain sampling schemes. This allows for efficient computation of zero-variance value estimates even along sampled trajectories.
\end{abstract}

\section{Introduction}

Multi-agent strategic interactions are often modeled as \emph{extensive-form games (EFGs)}, a game tree representation that allows for hidden information, stochastic outcomes, and sequential interactions. Research on solving EFGs has been driven by the experimental domain of poker games, in which the \emph{Counterfactual Regret Minimization (CFR)} algorithm \citep{Zinkevich07} has been the basis of several breakthroughs. Approaches incorporating CFR have been used to essentially solve one nontrivial poker game \citep{Bowling15}, and to beat human professionals in another \citep{Moravcik17, Brown18}.

CFR is in essence a policy improvement algorithm that iteratively evaluates and improves a strategy for playing an EFG. As part of this process, it must walk the entire game tree on every iteration. However, many games have prohibitively large trees when represented as EFGs. For example, many commonly played poker games have more possible game states than there are atoms in the universe \citep{Johanson13}. In such cases, performing even a single iteration of traditional CFR is impossible.

The prohibitive cost of CFR iterations is the motivation for \emph{Monte Carlo Counterfactual Regret Minimization (MCCFR)}, which samples trajectories to walk through the tree to allow for significantly faster iterations  \citep{Lanctot09}. Additionally, while CFR spends equal time updating every game state, the sampling scheme of MCCFR can be altered to target updates to parts of the game that are more critical or more difficult to learn \citep{Gibson12, Gibson12b}. As a trade-off for these benefits, MCCFR requires more iterations to converge due to the variance of sampled values.

In the Reinforcement Learning (RL) community, the topic of variance reduction in sampling algorithms has been extensively studied. In particular, baseline functions that estimate state values are typically used within policy gradient methods to decrease the variance of value estimates along sampled trajectories \citep{Williams92, Greensmith04, Bhatnagar09, Schulman16}. Recent work by \citet{Schmid19} has adapted these ideas to variance reduction in MCCFR, resulting in the VR-MCCFR algorithm.

In this work, we generalize and extend the ideas of Schmid et al. We introduce a framework for variance reduction of sampled values in EFGs by use of state-action baseline functions. We show that the VR-MCCFR is a specific application of our baseline framework that unnecessarily generalizes across dissimilar states. We introduce alternative baseline functions that take advantage of our access to the full hidden state during training, avoiding this generalization. Empirically, our new baselines result in significantly reduced variance and faster convergence than VR-MCCFR.

Schmid et al. also discuss the idea of an oracle baseline that provably minimizes variance, but is impractical to compute. We introduce a \emph{predictive baseline} that estimates this oracle value and can be efficiently computed. We show that under certain sampling schemes, the predictive baseline exactly tracks the true oracle value, thus provably computing zero-variance sampled values. For the first time, this allows for exact CFR updates to be performed along sampled trajectories.

\section{Background}\label{sec:bg}

An \emph{extensive-form game (EFG)} \citep{Osborne94} is a game tree, formally defined by a tuple $\langle \Players, \Histories, \Toact, \chances{\strategy}, \utility, \Isets \rangle$. $\Players$ is a finite set of players. $\Histories$ is a set of \emph{histories}, where each history is a sequence of \emph{actions} and corresponds to a vertex of the tree. For $\hist,\hist' \in \Histories$, we write $\hist \sqsubseteq \hist'$ if $\hist$ is a prefix of $\hist'$. The set of actions available at $\hist \in \Histories$ that lead to a successor history $(\hist\act) \in \Histories$ is denoted $\Actions(\hist)$. Histories with no successors are \emph{terminal histories} $\Terminals \subseteq \Histories$. $\Toact \colon \Histories \setminus \Terminals \to \Players \cup \{\chance\}$ maps each history to the player that chooses the next action, where $\chance$ is the \emph{chance} player that acts according to the defined distribution $\chances{\strategy}(\hist) \in \Delta_{\Actions(\hist)}$, where $\Delta_{\Actions(\hist)}$ is the set of probability distributions over $\Actions(\hist)$. The \emph{utility function} $\utility \colon \Players \times \Terminals \to \Reals$ assigns a value to each terminal history for each player.

For each player $\player \in \Players$, the collection of \emph{(augmented) information sets} $\my{\Isets} \in \Isets$ is a partition of the histories $\Histories$.\footnote{Augmented information sets were introduced by \citet{Burch14}.} Player $\player$ does not observe the true history $\hist$, but only the information set $\my{\Iset}(\hist)$. Necessarily, this means that $\Actions(\hist) = \Actions(\hist')$ if $\Iset_{\Toact(\hist)}(\hist) = \Iset_{\Toact(\hist)}(\hist')$, which we then denote $\Actions(\Iset)$.

Each player selects actions according to a \emph{(behavioral) strategy} that maps each information set $\Iset \in \my{\Isets}$ where $\Toact(\Iset) = \player$ to a distribution over actions, $\my{\strategy}(\Iset) \in \Delta_{\Actions(\Iset)}$. The probability of taking a specific action at a history is $\strategy_{\Toact(\hist)}(\hist, \act) = \strategy_{\Toact(\hist)}(\Iset(\hist), \act)$. A \emph{strategy profile}, $\strategy = \{\my{\strategy} | \player \in \Players\}$, specifies a strategy for each player. The \emph{reach probability} of a history $\hist$ is $\reach{\strategy}(\hist) = \prod_{(\hist'\act) \sqsubseteq \hist} \strategy_{\Toact(\hist')}(\hist', \act)$. This product can be decomposed as $\reach{\strategy}(\hist) = \my{\reach{\my{\strategy}}}(\hist) \opp{\reach{\opp{\strategy}}}(\hist)$, where the first term contains the actions of player $i$, and the second contains the actions of other players and chance. We also write $\reach{\strategy}(\hist, \hist')$ for the probability of reaching $\hist'$ from $\hist$, defined to be 0 if $\hist \not\sqsubseteq \hist'$. A strategy profile defines an expected utility for each player as $\my{\utility}(\strategy) = \my{\utility}(\my{\strategy},\opp{\strategy}) = \sum_{\term \in \Terminals} \reach{\strategy}(\term) \my{\utility}(\term)$.

In this work, we consider two-player zero-sum EFGs, in which $\Players = \{1,2\}$ and $\utility(\term) \coloneqq \my{\utility}(\term) = {-\opp{\utility}(\term)}$. We also assume that the information sets satisfy \emph{perfect recall}, which requires that players not forget any information that they once observed. Mathematically, this means that two histories in the same information set $\Iset_\player$ must have the same sequence of past information sets and actions for player $\player$. All games played by humans exhibit perfect recall, and solving games without perfect recall is NP-hard. We write $\my{\Iset} \sqsubseteq \hist$ if there is any history $\hist' \in \my{\Iset}$ such that $\hist' \sqsubseteq \hist$, and we denote that history (unique by perfect recall) by $\my{\Iset}[\hist]$.

\subsection{Solving EFGs}

A common solution concept for EFGs is a \emph{Nash equilibrium}, in which no player has incentive to deviate from their specified strategy. We evaluate strategy profiles by their distance from equilibrium, as measured by \emph{exploitability}, which is the average expected loss against a worst-case opponent: $\text{exploit}(\strategy) = \nicefrac{1}{2} \max_{\strategy' \in \Strategies} (\utility_2(\strategy_1,\strategy_2') + \utility_1(\strategy_1', \strategy_2))$.

\emph{Counterfactual Regret Minimization (CFR)} is an algorithm for learning Nash equilibria in EFGs through iterative self play \citep{Zinkevich07}. For any $\hist \in \Histories$, let $\Terminals[\hist] = \{\term \in \Terminals \mid \hist \sqsubseteq \term\}$ be the set of terminal histories reachable from $\hist$, and define the history's expected utility as $\exputil{\hist}{\strategy} = \sum_{\term \in \Terminals[\hist]} \reach{\strategy}(\hist, \term) \utility(\term)$. For each information set $\Iset$ and action $\act \in \Actions(\Iset)$, CFR accumulates the \emph{counterfactual regret} of not choosing that action on previous iterations:
\begin{equation}
\regret{\time}(\Iset, \act) = \sum_{\hist \in \Iset} \reach{\strategy^\time}_{-\Toact(\hist)}(\hist) \left(\exputil{(\hist\act)}{\strategy} - \exputil{\hist}{\strategy}\right)
\qquad
\Regret{\Time}(\Iset, \act) = \sum_{\time=1}^\Time \regret{\time}(\Iset, \act)
\label{eq:cfr}
\end{equation}
The next strategy profile is then selected with \emph{regret matching}, which sets probabilities proportional to the positive regrets: $\strategy^{\Time+1}(\Iset, \act) \propto \max(\Regret{\Time}(\Iset,\act), 0)$. Defining the average strategy $\avgstrat^\Time$ such that $\avgstrat^\Time(\hist, \act) \propto \sum_{\time=1}^\Time \my{\reach{\strategy^\time}}(\hist) \my{\strategy^\time}(\hist, \act)$, CFR guarantees that $\text{exploit}(\avgstrat^\Time) \to 0$ as $\Time \to \infty$, thus converging to a Nash equilibrium.

The state-of-the-art \emph{CFR+} variant of CFR greedily zeroes all negative regrets on every iteration, replacing $\Regret{\time}$ with an accumulant $Q^\time$ recursively defined with $Q^0(\Iset, \act) = 0, Q^\time(\Iset, \act) = \max(Q^{\time-1}(\Iset, \act) + \regret{\time}(\Iset, \act), 0)$ \citep{Tammelin15}. It also alternates updates for each player, and uses linear averaging, which gives greater weight to more recent strategies.

CFR(+) requires a full walk of the game tree on each iteration, which can be a very costly operation on large games. \emph{Monte Carlo Counterfactual Regret Minimization (MCCFR)} avoids this cost by only updating along sampled trajectories. For simplicity,  we focus on the \emph{outcome sampling (OS)} variant of MCCFR \citep{Lanctot09}, though all results in this paper can be trivially extended to other MCCFR variants. On each iteration $\time$, a sampling strategy $\sampling{\time} \in \Strategies$ is used to sample a single terminal history $\term^\time \sim \reach{\sampling{\time}}$. A sampled utility is then calculated recursively for each prefix of $\term^\time$ as
\begin{equation}
\samputil{\hist, \act}{\strategy^\time}{\term^\time} = \frac{\indic((\hist\act) \sqsubseteq \term^\time)}{\sampling{\time}(\hist, \act)} \samputil{(\hist\act)}{\strategy^\time}{\term^\time}
\qquad
\samputil{\hist}{\strategy^\time}{\term^\time} = \sum_{\act \in \Actions(\hist)} \strategy^\time(\hist, \act) \samputil{\hist, \act}{\strategy^\time}{\term^\time}
\end{equation}
where $\indic$ is the indicator function and $\samputil{\term^\time}{\strategy^\time}{\term^\time} = \utility(\term^\time)$. For any $\hist \sqsubseteq \term^\time$, the sampled value $\samputil{\hist, \act}{\strategy^\time}{\term^\time}$ is an unbiased estimate of the expected utility $\exputil{(\hist\act)}{\strategy^\time}$, whether $\act$ is sampled or not. These sampled values are used to calculate a sample of the counterfactual regret:
\begin{equation}
\sampreg{\time}{\Iset, \act}{\term^\time} = \sum_{\hist \in \Iset} \frac{\reach{\strategy^\time}_{-\Toact(\hist)}(\hist)}{\reach{\sampling{\time}}(\hist)} \left(\samputil{\hist, \act}{\strategy^\time}{\term^\time} - \samputil{\hist}{\strategy^\time}{\term^\time}\right)
\label{eq:sampledregret}
\end{equation}
This gives an unbiased sample of the counterfactual regret $\regret{\time}(\Iset, \act)$ for all $\Iset \in \Isets$, which is then used to perform unbiased CFR updates. As long as as the sampling strategies satisfy $\reach{\sampling{\time}}(\term) > 0$ for all $\term \in \Terminals$, MCCFR guarantees that $\text{exploit}(\avgstrat^\Time) \to 0$ with high probability as $\Time \to \infty$, thus converging to a Nash equilibrium. However, the rate of convergence depends on the variance of $\sampreg{\time}{\Iset, \act}{\term^\time}$ \citep{Gibson12}.

\section{Baseline framework for EFGs}
\label{sec:framework}

We now introduce a method for calculating unbiased estimates of utilities in EFGs that has lower variance than the sampled utilities $\samputil{\hist, \act}{\strategy^\time}{\term^\time}$ defined above.  We do this using \emph{baseline functions}, which estimate the expected utility of actions in the game. We will describe specific examples of such functions in Section~\ref{sec:baselines}; for now, we assume the existence of some function $\baseline^\time \colon \Histories \times \Actions \to \Reals$ such that $\baseline^\time(\hist, \act)$ in some way approximates $\exputil{(\hist\act)}{\strategy^\time}$. We define a baseline-corrected sampled utility as
\begin{align}
\baseutil{\hist, \act}{\strategy^\time}{\term^\time} &= \frac{\indic((\hist\act) \sqsubseteq \term^\time)}{\sampling{\time}(\hist, \act)} \left(\baseutil{(\hist\act)}{\strategy^\time}{\term^\time} - \baseline^\time(\hist, \act)\right) + \baseline^\time(\hist, \act)\label{eq:baseline}\\
\baseutil{\hist}{\strategy^\time}{\term^\time} &= \sum_{\act \in \Actions(\hist)} \strategy^\time(\hist, \act) \baseutil{\hist, \act}{\strategy^\time}{\term^\time}
\end{align}

Equation (\ref{eq:baseline}) comes from the application of a \emph{control variate}, in which we lower the variance of a random variable ($X = \frac{\indic((\hist\act) \sqsubseteq \term^\time)}{\sampling{\time}(\hist, \act)} \baseutil{(\hist\act)}{\strategy^\time}{\term^\time}$) by subtracting another random variable ($Y = \frac{\indic((\hist\act) \sqsubseteq \term^\time)}{\sampling{\time}(\hist, \act)} \baseline^\time(\hist, \act)$) and adding its known expectation ($\Exp{Y} = \baseline^\time(\hist, \act)$), thus keeping the resulting estimate unbiased. If $X$ and $Y$ are correlated, then this estimate will have lower variance than $X$ itself. Because $\baseutil{(\hist\act)}{\strategy^\time}{\term^\time}$ is defined recursively, its computation includes the application of independent control variates at every action taken between $\hist$ and $\term^\time$.

These estimates are unbiased and, if the baseline function is chosen well, have low variance:

\begin{thm}
For any $\hist \sqsubseteq \term^\time$ and any $\act \in \Actions(\hist)$, the baseline-corrected utilities satisfy
\begin{equation*}
\Exp[\term^\time]{\baseutil{\hist, \act}{\strategy^\time}{\term^\time}|\term^\time \sqsupseteq \hist} = \exputil{(\hist\act)}{\strategy^\time}
\qquad
\Exp[\term^\time]{\baseutil{\hist}{\strategy^\time}{\term^\time}|\term^\time \sqsupseteq \hist} = \exputil{\hist}{\strategy^\time}
\end{equation*}
\label{thm:unbiased}
\end{thm}

\begin{thm} Assume that we have a baseline that satisfies $\baseline^\time(\hist,\time) = \exputil{(\hist\act)}{\strategy^\time}$ for all $\hist \in \Histories$, $\act \in \Actions(\hist)$. Then for any $\hist, \act, \term^\time$,
\begin{equation*}
\Var[\term^\time]{\baseutil{\hist, \act}{\strategy^\time}{\term^\time} | \term^\time \sqsupseteq \hist} = 0
\end{equation*}
\label{thm:variance}
\end{thm}

All proofs are given in \iftoggle{appendix}{the appendix.}{the supplementary materials.} Theorem~\ref{thm:unbiased} show that we can use $\baseutil{\hist, \act}{\strategy^\time}{\term^\time}$ in place of $\samputil{\hist, \act}{\strategy^\time}{\term^\time}$ in equation~\ref{eq:sampledregret} and maintain the convergence guarantees of MCCFR. Theorem~\ref{thm:variance} shows that an ideal baseline eliminates all variance in the MCCFR update. By choosing our baseline well, we decrease the MCCFR variance and speed up its convergence. Pseudocode for MCCFR with baseline-corrected values is given in \iftoggle{appendix}{Appendix~\ref{sec:pseudocode}}{the supplementary materials}.

Although we focus on using our baseline-corrected samples in MCCFR, nothing in the value definition is particular to that algorithm. In fact, a lower variance estimate of sampled utilities is useful in any algorithm that performs iterative training using sampled trajectories. Examples of such algorithms include policy gradient methods \citep{Srinivasan18} and stochastic first-order methods \citep{Kroer15}.\todo{Other examples?}\todo{This paragraph here or in related work?}

\section{Baselines for EFGs}
\label{sec:baselines}

In this section we propose several baseline functions for use during iterative training. Theorem~\ref{thm:variance} shows that we can minimize variance by choosing a baseline function $\baseline^\time$ such that $\baseline^\time(\hist, \act) \approx \exputil{(\hist\act)}{\strategy^\time}$.

\paragraph{No baseline.}

We begin by examining MCCFR under its original definition, where no baseline function is used. We note that when we run baseline-corrected MCCFR with a static choice of $\baseline^\time(\hist, \act) = 0$ for all $\hist, \act$, the operation of the algorithm is identical to MCCFR. Thus, opting to not use a baseline is, in itself, a choice of a very particular baseline.

Using $\baseline^\time(\hist, \act) = 0$ might seem like a reasonable choice when we expect the game's payouts to be balanced between the players. However, even when the overall expected utility $\utility(\strategy)$ is very close to 0, there will usually be particular histories with high magnitude expected utility $\exputil{\hist}{\strategy}$. For example, in poker games, the expected utility of a history is heavily biased toward the player who has been dealt better cards, even if these biases cancel out when considered across all histories. In fact, often there is no strategy profile at all that satisfies $\exputil{(\hist\act)}{\strategy} = 0$, which makes $\baseline^\time(\hist, \act) = 0$ a poor choice in regards to the ideal criteria $\baseline^\time(\hist, \act) \approx \exputil{(\hist\act)}{\strategy^\time}$. An example game where a zero baseline performs very poorly is explored in Section~\ref{sec:results}.

\paragraph{Static strategy baseline.}

The simplest way to ensure that the baseline function does correspond to an actual strategy is to choose a static, known strategy profile $\strategy^\baseline \in \Strategies$ and let $\baseline^\time(\hist, \act) = \exputil{(\hist\act)}{\strategy^\baseline}$ for each time $t$. Once the strategy is chosen, the baseline values only need to be computed once and stored. In general this requires a full walk of the game tree, but it is sometimes possible to take advantage of the structure of the game to greatly reduce this cost. For an example, see Section~\ref{sec:results}.

\paragraph{Learned history baseline.}

Using a static strategy for our baseline ensures that it corresponds to some expected utility, but it fails to take advantage of the iterative nature of MCCFR. In particular, when attempting to estimate $\exputil{(\hist\act)}{\strategy^\time}$, we have access to all past samples $\baseutil{(\hist\act)}{\strategy^\tau}{\term^\tau}$ for $\tau < \time$. Because the strategy is changed incrementally, we might expect the expected utility to change slowly and for these to be reasonable samples of the utility at time $t$ as well.

Define $\Timesteps^{\hist\act}(\time) = \{ \tau < \time \mid (\hist\act) \sqsubseteq \term^\tau \}$ to be the set of timesteps on which $(\hist\act)$ was sampled, and denote the $j$th such timestep as $\tau_j$. We define the \emph{learned history baseline} as
\begin{equation}
\baseline^\time(\hist, \act) = \sum_{j=1}^{|\Timesteps^{\hist\act}(\time)|} \weight{j}\baseutil{(\hist\act)}{\strategy^{\tau_j}}{\term^{\tau_j}}
\end{equation}
where $(\weight{j})_{j=1}^{|\Timesteps^{\hist\act}(\time)|}$ is a sequence of weights satisfying $\sum_{j=1}^{|\Timesteps^{\hist\act}(\time)|} \weight{j} = 1$. Possible weighting choices include simple averaging, where $\weight{j}= 1/|\Timesteps^{\hist\act}(\time)|$, and exponentially-decaying averaging, where $\weight{j} = \alpha(1-\alpha)^{|\Timesteps^{\hist\act}(\time)| - j}$ for some $\alpha \in (0,1]$. In either case, the baseline can be efficiently updated online by tracking the weighted sum and the number of times that $(\hist\act)$ has been sampled.

\paragraph{Learned infoset baseline.}

The learned history baseline is very similar to the VR-MCCFR baseline defined by \citet{Schmid19}. The principle difference is that the VR-MCCFR baseline tracks values for each information set, rather than for each history; we thus refer to it as the \emph{learned infoset baseline}. This baseline also updates values for each player separately, based on their own information sets. This can be accomplished by tracking separate values for each player throughout the tree walk, or by running MCCFR with alternating updates, where only one player's regrets are updated on each tree walk. The VR-MCCFR baseline can be defined in our framework as
\begin{equation}
\baseline^\time(\hist, \act) = \baseline^\time(\my{\Iset}(\hist), \act)\qquad\text{where}\qquad
\baseline^\time(\my{\Iset}, \act) = \sum_{j=1}^{|\Timesteps^{\my{\Iset}\act}(\time)|} \weight{j}\baseutil{(\my{\Iset}[\term^{\tau_j}]\act)}{\strategy^{\tau_j}}{\term^{\tau_j}}
\end{equation}
where $\player$ is the player being updated, $\Timesteps^{\my{\Iset}\act}(\time)$ is the set of timesteps on which $(\hist'\act)$ was sampled for any $\hist' \in \my{\Iset}$, and $\tau_j$ is $j$th such timestep. Following Schmid et al. we consider both simple averaging and exponentially-decaying averaging for selecting the weights $\weight{j}$.

\paragraph{Predictive baseline.}

Our last baseline takes advantage of the recursive nature of the MCCFR update. On each iteration, each history along the sampled trajectory is evaluated and updated in depth-first order. Thus when the update of history $\hist \sqsubseteq \term^\time$ is complete and the value is returned , we have already calculated the next regrets $\Regret{\time+1}(\Iset(\hist'),\cdot)$ for all $\hist'$ such that $\hist \sqsubseteq \hist' \sqsubseteq \term^\time$. These values will be the input to the regret matching procedure on the next iteration, computing $\strategy^{\time+1}(\hist', \cdot)$ at these histories. Thus we can immediately compute this next strategy, and using the already sampled trajectory, compute an estimate of the strategy's utility as $\baseutil{(\hist\act)}{\strategy^{\time+1}}{\term^\time}$. This is an unbiased sample of the expected utility $\exputil{(\hist\act)}{\strategy^{\time+1}}$, which is our target value for the next baseline $\baseline^{\time+1}(\hist, \act)$. We thus use this sample to update the baseline:
\begin{equation}
\baseline^{\time+1}(\hist, \act) = \begin{cases}\baseutil{(\hist\act)}{\strategy^{\time+1}}{\term^\time}\qquad&\text{if~}(\hist\act) \sqsubseteq \term^\time\\
\baseline^\time(\hist, \act)&\text{otherwise}\end{cases}
\end{equation}

The computation for this update can be done efficiently by a simple change to MCCFR. In MCCFR, we compute $\baseutil{\hist}{\strategy^{\time}}{\term^\time}$ at each step by using $\strategy^\time$ to weight recursively-computed action values. In MCCFR with predictive baseline, after updating the regrets at $\hist$, we use a second regret matching computation to compute $\strategy^{\time+1}(\hist, \cdot)$. We use this strategy to weight a second set of recursively-computed action values to compute $\baseutil{\hist}{\strategy^{\time+1}}{\term^\time}$. When we walk back up the tree, we return both of the values $\baseutil{\hist}{\strategy^{\time}}{\term^\time}$ and $\baseutil{\hist}{\strategy^{\time+1}}{\term^\time}$, allowing this recursion to continue. The predictive value $\baseutil{\hist}{\strategy^{\time+1}}{\term^\time}$ is only used for updating the baseline function. These changes do not modify the asymptotic time complexity of MCCFR. Pseudocode is given in \iftoggle{appendix}{Appendix~\ref{sec:pseudocode}}{the supplementary materials}.

\todo{Talk about combining baselines?}

\section{Experimental comparison}
\label{sec:results}

We run our experiments using a commodity desktop machine in Leduc hold'em \citep{Southey05}, a small poker game commonly used as a benchmark in games research\footnote{An open source implementation of CFR+ and Leduc hold'em is available from the University of Alberta \citep{CPRGcode}.}. We compare the effect of the various baselines on the MCCFR convergence rate. Our experiments use the regret zeroing and linear averaging of CFR+, as these improve convergence when combined with any baseline. For the static strategy baseline, we use the ``always call'' strategy, which matches the opponent's bets and makes no bets of its own. Expected utility under this strategy is determined by the current size of the pot, which is measurable at run time, and the winning chance of each player's cards. Before training, we measure and store these odds for all possible sets of cards, which is significantly smaller than the size of the full game. For both of the learned baselines, we use simple averaging as it was found to work best in preliminary experiments.

We run experiments with two sampling strategies. The first is uniform sampling, in which $\sampling{\time}(\hist, \act) = 1/{|\Actions(\hist)|}$. The second is opponent on-policy sampling, depends on the player $\player$ being updated: we sample uniformly ($\sampling{\time}(\hist, \act) = 1/{|\Actions(\hist)|}$) at histories $\hist$ where $\Toact(\hist) = \player$, and sample on-policy ($\sampling{\time}(\hist, \act) = \strategy^\time(\hist, \act)$) otherwise. For consistency, we use alternating updates for both schemes.

\begin{figure*}[htb]
\centering
\begin{subfigure}[t]{0.34\linewidth}
\centering
  \includegraphics[height=2.1in]{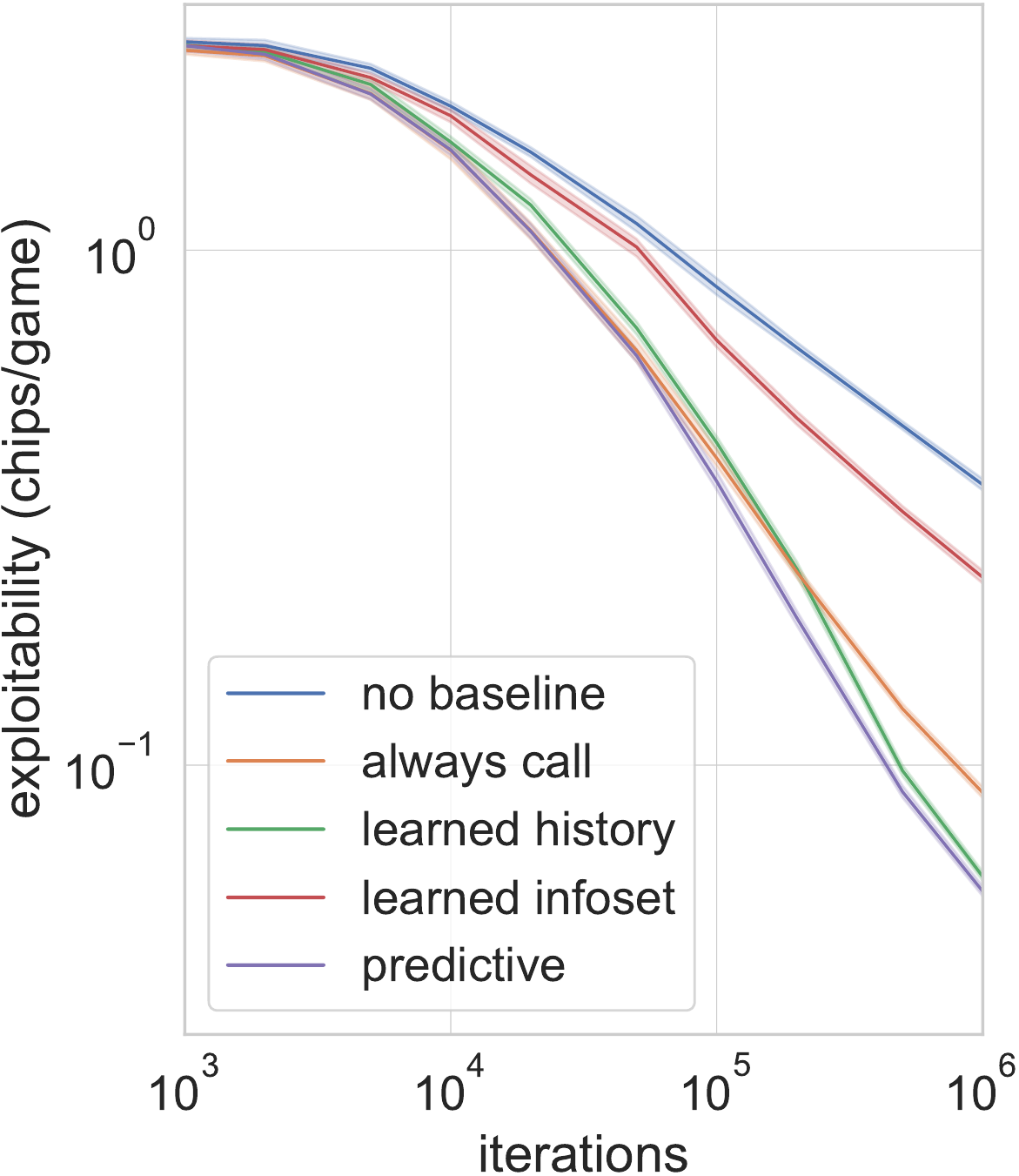}
  \caption{Uniform sampling}
  \label{fig:osexp}
\end{subfigure}%
\begin{subfigure}[t]{0.325\linewidth}
\centering
  \includegraphics[height=2.1in]{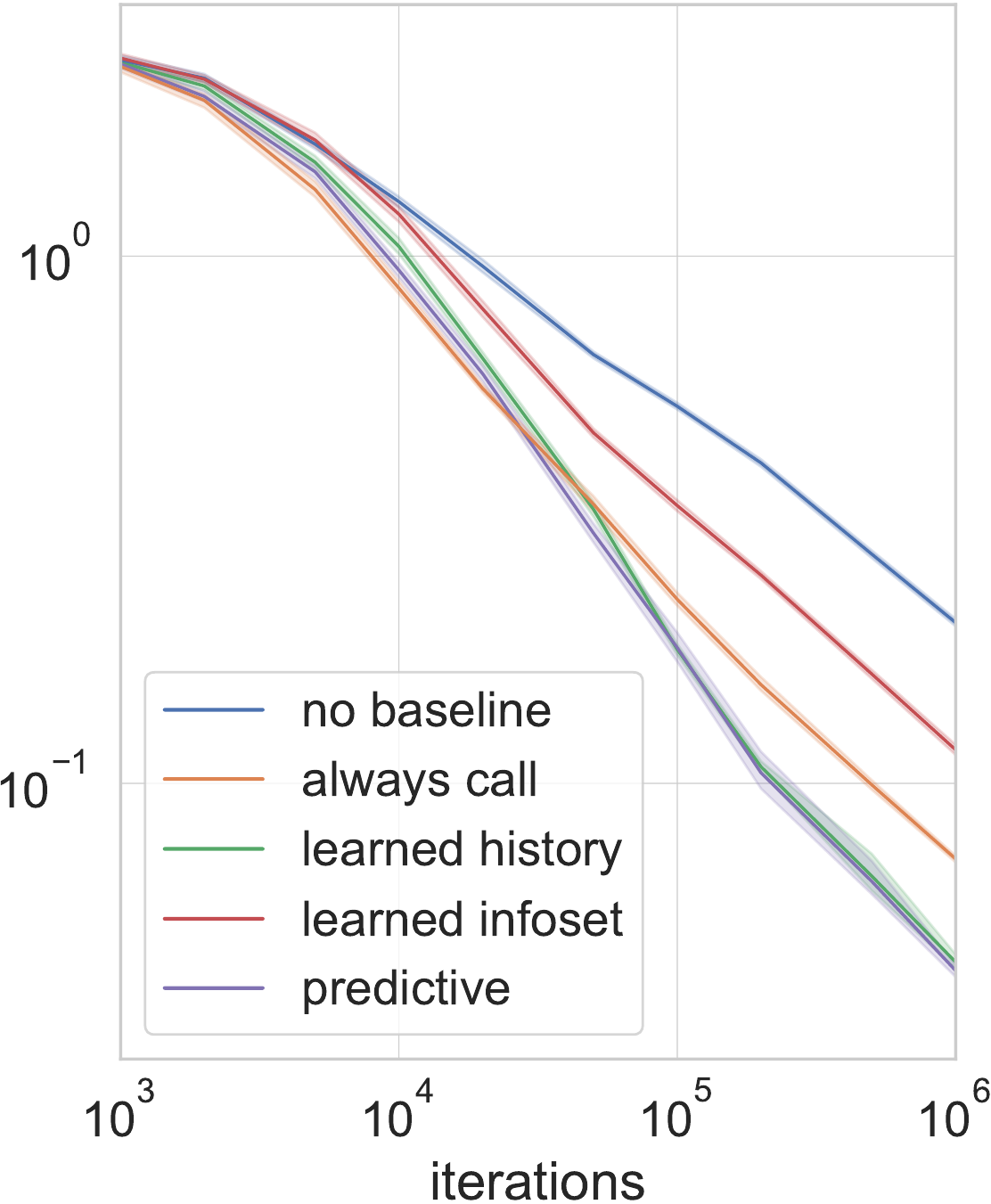}
  \caption{Opponent on-policy sampling}
  \label{fig:oposexp}
\end{subfigure}%
\begin{subfigure}[t]{0.325\linewidth}
\centering
  \includegraphics[height=2.1in]{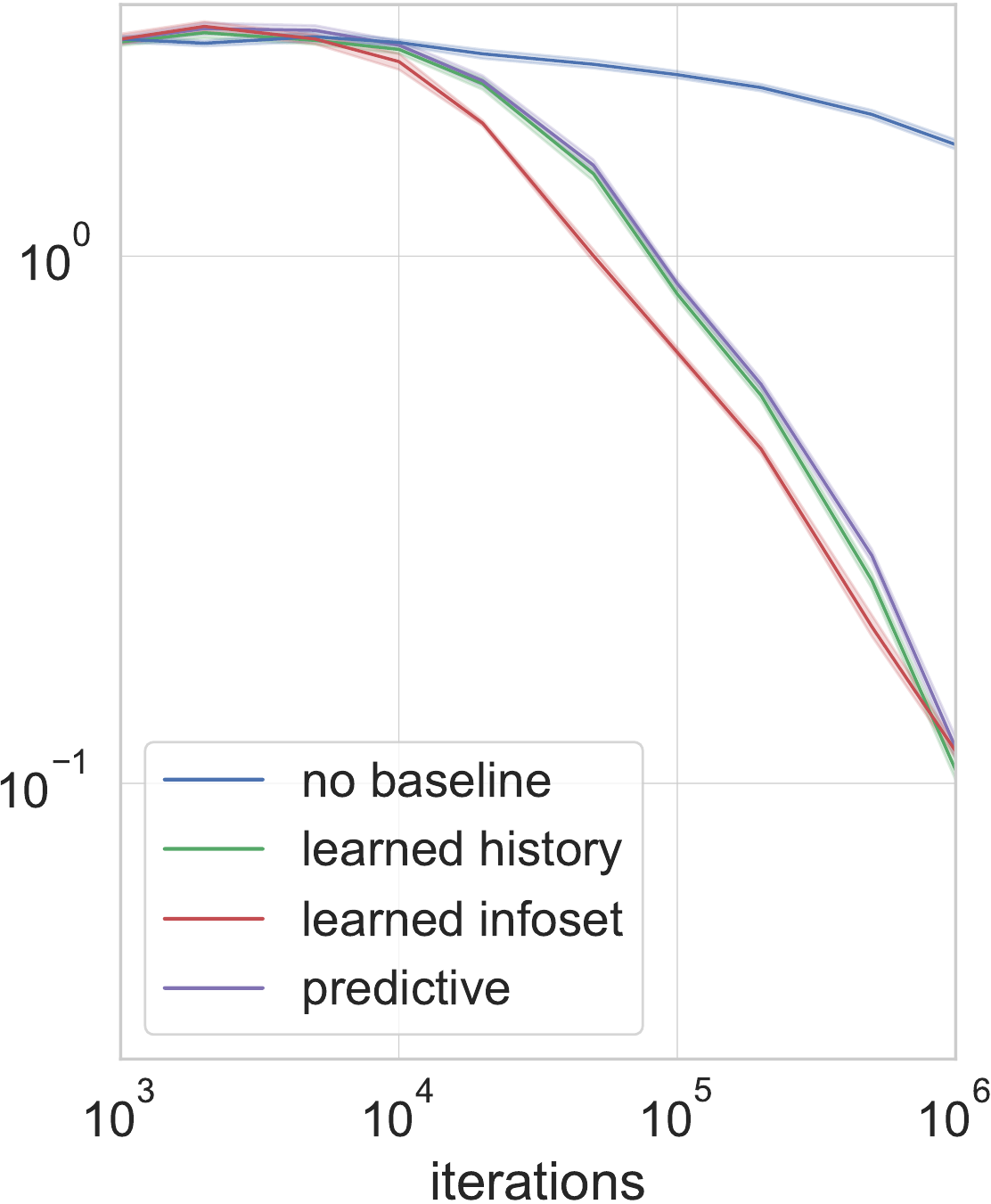}
  \caption{Shifted utilities}
  \label{fig:shift}
\end{subfigure}%
\caption{Log-log plots of convergence of MCCFR strategies with various baselines. \textbf{(a)} and \textbf{(b)} Leduc with different MCCFR sampling schemes. \textbf{(c)} Leduc with utilities shifted by 100 and opponent on-policy sampling.}
\end{figure*} 

Figures~\ref{fig:osexp}~and~\ref{fig:oposexp} show the convergence of MCCFR with the various baselines, as measured by exploitability (recall that exploitability converges to zero). All results in this paper are averaged over 20 runs, with 95\% confidence intervals shown as error bands (often too narrow to be visible). With uniform sampling, the learned infoset (VR-MCCFR) baseline improves modestly on using no baseline at all, while the other three baselines achieve a significant improvement on top of that. With opponent on-policy sampling, the gap is smaller, but the learned infoset baseline is still noticeably worse than the other options.

Many true expected values in Leduc are very close to zero, making MCCFR without a baseline (i.e. $\baseline^\time(\hist, \act) = 0$) better than it might otherwise be. To demonstrate the necessity of a baseline in some games, we ran MCCFR in a modified Leduc game where player 2 always transfers 100 chips to player 1 after every game. This utility change is independent of the player's actions, so it doesn't strategically change the game. However, it means that 0 is now an extremely inaccurate value estimate for all histories. Figure~\ref{fig:shift} shows convergence in Leduc with shifted utilities. The always call baseline is omitted, as the results would be identical to those in Figure~\ref{fig:oposexp}. Here we see that using any baseline at all provides a significant advantage over not using a baseline, due to the ability to adapt to the shifted utilities. We also see that the learned infoset baseline performs relatively well early on in this setting, because it generalizes across histories.

\section{Public Outcome Sampling}

Although the results in Section~\ref{sec:results} show large gains in convergence speed when using baselines with MCCFR, the magnitudes are not as large as those shown with the VR-MCCFR baseline by \citet{Schmid19}. This is because their experiments use a "vectorized" form of MCCFR, which avoids the sampling of individual histories within information sets. Instead, they track a vector of values on each iteration, one for each possible true history given the player's observed information set. Schmid et al. do not formally define their algorithm. We refer to it as \emph{Public Outcome Sampling (POS)} as the algorithm samples any actions that are publicly visible to both players, while exhaustively considering all possible private states. We give a full formal definition of POS in \iftoggle{appendix}{Appendix~\ref{sec:pubtrees}}{the supplementary materials}.

\subsection{Baselines in POS}

In MCCFR with POS, we still use action baselines $\baseline^\time(\hist, \act)$ with the ideal baseline values being $\baseline^\time(\hist,\act) = \exputil{(\hist\act)}{\strategy^\time}$. Thus the baselines in Section~\ref{sec:baselines} apply to this setting as well.

For the learned infoset baseline, we have more information available to us than in the OS case. This is because when POS samples some history-action pair $\hist,\act$, it also samples every pair $\hist',\act$ for $\hist' \in \Iset(\hist)$. Thus, rather, than using one sampled history value to update the baseline, we use a weighted sample of all of the history values. Following Schmid et al., we weight the baseline values
\begin{equation*}
\baseline^\time(\hist, \act) = \baseline^\time(\my{\Iset}(\hist), \act)\qquad\text{where}\qquad
\baseline^\time(\my{\Iset}, \act) = \sum_{j=1}^{|\Timesteps^{\my{\Iset}\act}(\time)|} \weight{j}\frac{\sum_{\hist' \in \my{\Iset}}\opp{\reach{\strategy^{\tau_j}}}(\hist')\baseutil{(\hist'\act)}{\strategy^{\tau_j}}{\term^{\tau_j}}}{\sum_{\hist' \in \my{\Iset}}\opp{\reach{\strategy^{\tau_j}}}(\hist')}.
\end{equation*}
This is the same relative weighting given to each history when calculating the counterfactual regret.

\paragraph{Zero-variance baseline.}

POS also has implications for the predictive baseline. In fact, we can guarantee that after every outcome of the game has been sampled, the predictive baseline will have learned the true value of the current strategy. For time $\time$, let $\Terminals^\time$ be the set of sampled terminal histories (consistent with a public outcome), and let $\text{samp}^\time(\hist)$ be the event that $\hist$ is sampled on way to $\Terminals^\time$.

\begin{thm}
If each of the terminal states $\Terminals[\hist]$ reachable from history $\hist \in \Histories$ has been sampled at least once under public outcome sampling \textup{($\Terminals[\hist] \subseteq \bigcup_{\tau < \time} \Terminals^\tau$)}, then the predictive baseline satisfies
\begin{equation*}
\baseline^\time(\hist,\act) = \exputil{(\hist\act)}{\strategy^\time} \textup{~~and~~} \Var[\Terminals^\time]{\baseutil{\hist}{\strategy^\time}{\Terminals^\time}|\mathrm{samp}^\time(\hist)} = 0\qquad \forall \act \in \Actions(\hist)
\end{equation*}%
\label{thm:zerovar}%
\end{thm}%
The key idea behind the proof is that POS ensures that the baseline is updated at a history if and only if the expected value of the history changes. The full proof is in \iftoggle{appendix}{Appendix~\ref{sec:zerovar}}{the supplementary materials}.

In order for the theorem to hold everywhere in the tree, all outcomes must be sampled, which could take a large number of iterations. An alternative is to guarantee that all outcomes are sampled during the early iterations of MCCFR. For example, one could do a full CFR tree walk on the very first iteration, and then sample on subsequent iterations. Alternatively, we can ensure the theorem always holds with smart initialization of the baseline. When there are no regrets accumulated, MCCFR uses an arbitrary strategy. If we have some strategy with known expected values throughout the tree, we can use this strategy as the default MCCFR strategy and initialize the baseline values to the strategy's expected values. Either option guarantees that all regret updates will use zero-variance values.

\subsection{POS results}

As in Section~\ref{sec:results}, we run experiments in Leduc and use CFR+ updates. For the learned baselines, we use exponentially-decaying averaging with $\alpha=0.5$, which preliminary experiments found to outperform simple averaging when combined with POS. For simplicity and consistency with the experiments of \citet{Schmid19}, we use uniform sampling and simultaneous updates.

\begin{figure*}[pth]
\centering
\begin{subfigure}[t]{0.46\linewidth}
\centering
  \includegraphics[width=\linewidth]{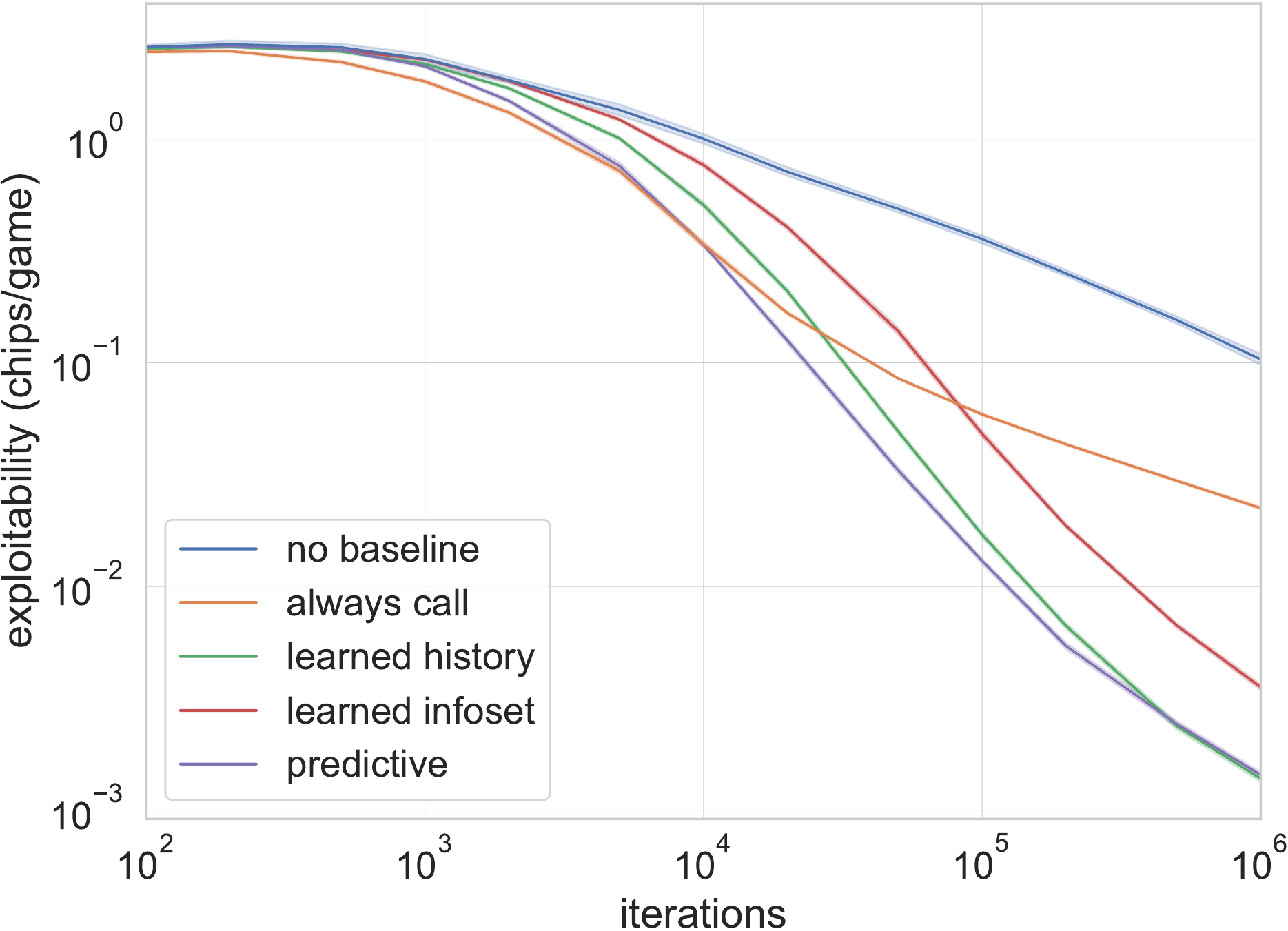}
  \caption{Exploitability}
  \label{fig:posexp}
\end{subfigure}\quad%
\begin{subfigure}[t]{0.46\linewidth}
\centering
  \includegraphics[width=\linewidth]{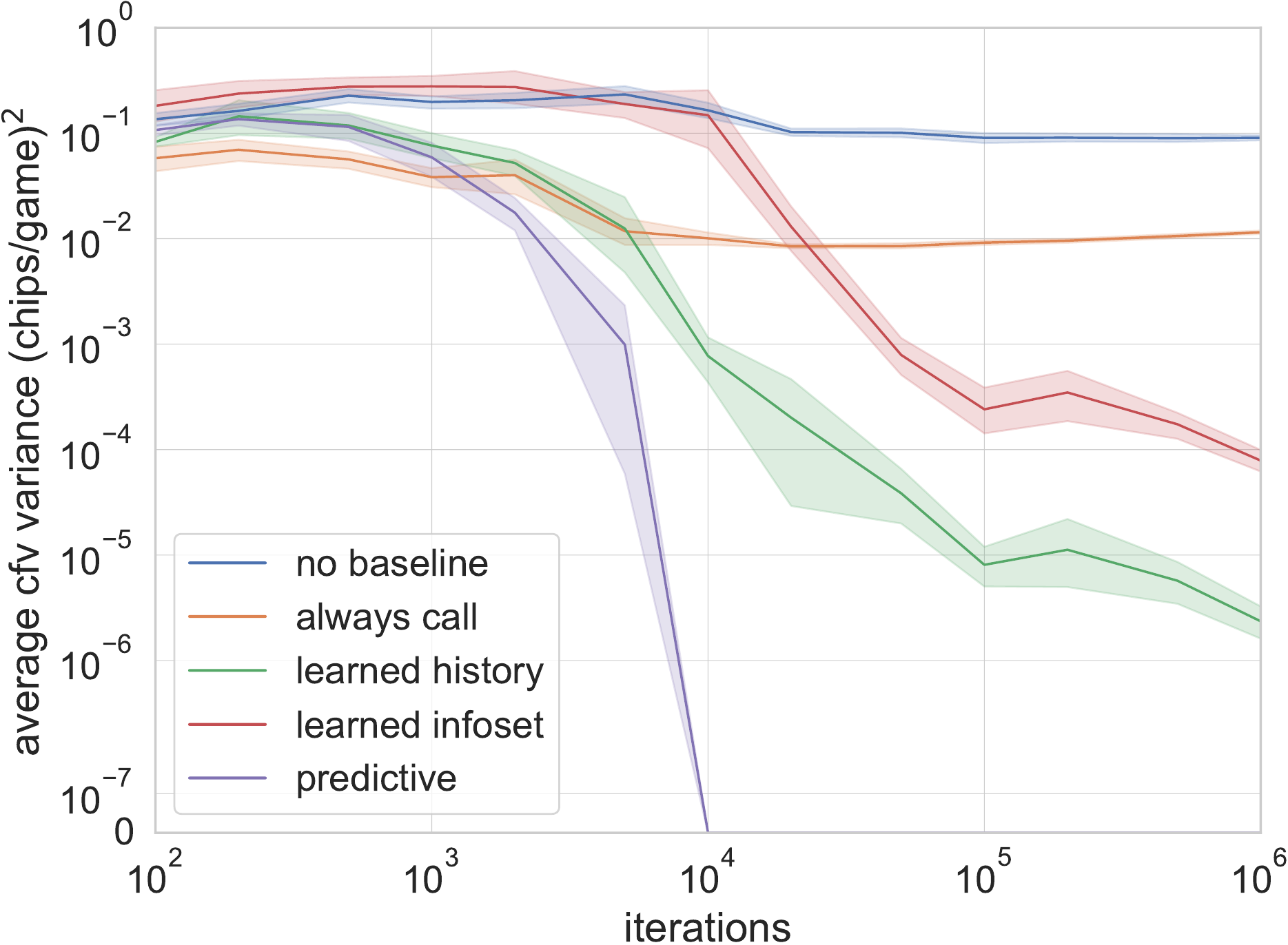}
  \caption{Counterfactual value variance}
  \label{fig:var}
\end{subfigure}%
\caption{Log-log plots of POS MCCFR strategies with various baselines. \textbf{(a)} Convergence as measured by exploitability. \textbf{(b)} Empirical variance of counterfactual values (cfvs).}
\end{figure*}

Figure~\ref{fig:posexp} compares the baselines' effects on POS MCCFR. We find that using any baseline provides a significant improvement on using no baseline. The always call baseline performs well early but tales off as it doesn't learn during training. Even with POS, where we always see an entire information set at a time, the learned infoset baseline (VR-MCCFR) is significantly outperformed by the learned history and predictive baselines. This is likely because the learned infoset baseline has to learn the relative weighting between histories in an infoset, while the other baselines always use the current strategy to weight the learned values. Finally, we observe that the predictive baseline has a small, but statistically significant, advantage over the learned history baseline in early iterations.

In addition, we compare the baselines by directly measuring their variance. We measure the variance of the counterfactual value $\sum_{\hist \in \Iset} \reach{\strategy^\time}_{-\Toact(\hist)}(\hist)\baseutil{(\hist\act)}{\strategy^\time}{\term^\time}$ for each pair $\Iset, \act$, and we average across all such pairs. Full details are in \iftoggle{appendix}{Appendix~\ref{sec:empvar}}{the supplementary materials}. Results are shown in Figure~\ref{fig:var}. We see that using no baseline results in high and relatively steady variance of counterfactual values. Using the always call baselines also results in steady variance, as nothing is learned, but at approximately an order of magnitude lower than no baseline. Variance with the other baselines improves over time, as the baseline becomes more accurate. The learned history baseline mirrors the learned infoset baseline, but with more than an order of magnitude reduction in variance. The predictive baseline is best of all, and in fact we see Theorem~\ref{thm:zerovar} in action as the variance drops to zero.

\section{Related Work}\todo{Do we need to talk about VR-MCCFR again here?}

As discussed in the introduction, the use of baseline functions has a long history in RL. Typically these approaches have used state value baselines, with some recent exceptions \citep{Liu18, Wu18}. \citet{Tucker18} suggest an explanation for this by isolating the variance terms that come from sampling an immediate action and from sampling the rest of a trajectory. Typical RL baselines only reduce the action variance, so the additional benefit from using a state-action baseline is insignificant when compared to the trajectory variance. In our work, we apply a recursive baseline to reduce both the action and trajectory variances, meaning state-action baselines give a noticeable benefit.

In RL, the doubly-robust estimator \citep{Jiang16} has been used to reduce variance settings by the recursive application of control variates \citep{Thomas16}. Similarly, variance reduction in EFGs via recursive control variates is the basis of the advantage sum estimator \citep{Zinkevich07} and AIVAT \citep{Burch18}. All of these techniques construct control variates by evaluating a static policy or strategy, either on the true game or on a static model. In this sense they are equivalent to our static strategy baseline. However, to the best of our knowledge, these techniques have only been used for evaluation of static strategies, rather than for variance reduction during training. Our work extends the EFG techniques to the training domain; we believe that similar ideas can be used in RL, and this is an interesting avenue of future research.

Concurrent to this work, \citet{Zhou18} also suggested tracking true values of histories in a CFR variant, analogous to our predictive baseline. They use these values for truncating full tree walks, rather than for variance reduction along sampled trajectories. As such, they always initialize their values with a full tree walk, and don't examine gradually learning the values during training.

\section{Conclusion and Future Work}

In this work we introduced a new framework for variance reduction in EFGs through the application of a baseline value function. We demonstrated that the existing VR-MCCFR baseline can be described in our framework with a specific baseline function, and we introduced other baseline functions that significantly outperform it in practice. In addition, we introduced a predictive baseline and showed that it gives provably optimal performance under a sampling scheme that we formally define.

There are three sources of variance when performing sampled updates in EFGs. The first is from sampling trajectory values, the second from sampling individual histories within an information set that is being updated, and the third from sampling which information sets will be updated on a given iteration. By introducing MCCFR with POS, we provably eliminate the first two sources of variance: the first because we have a zero-variance baseline, and the second because we consider all histories within the information set. For the first time, this allows us to select the MCCFR sampling strategy $\sampling{\time}$ entirely on the basis of minimizing the third source of variance, by choosing the ``best'' information sets to update. Doing this in a principled way is an exciting avenue for future research.

Finally, we close by discussing function approximation. All of the baselines introduced in this paper require an amount of memory that scales with the size of the game tree. In contrast, baseline functions in RL typically use function approximation, requiring a much smaller number of parameters. Additionally, these functions generalize across states, which can allow for learning an accurate baseline function more quickly. The framework that we introduce in this work is completely compatible with function approximation, and combining the two is an area for future research.

\small
\bibliographystyle{plainnat}
\bibliography{refs}

\iftoggle{appendix}{
    \newpage
    \appendix
    \appendixpage
    \section{MCCFR with baseline-corrected values}
\label{sec:pseudocode}

Pseudocode for MCCFR with baseline-corrected values is given in Algorithm~\ref{alg:bmccfr}. Quantities of the form $\strategy^\time(\hist, \cdot)$ refer to the vector of all quantities $\strategy^\time(\hist, \act)$ for $\act \in \Actions(\hist)$. A version for the predictive baseline, which must calculate extra values, is given in Algorithm~\ref{alg:predmccfr}. Each of these algorithms has the same worst-case iteration complexity as MCCFR without baselines, namely $\BigO{d|\Actions_{\text{max}}|}$ where $d$ is the tree's depth and $|\Actions_{\text{max}}| = \max_\hist |\Actions(\hist)|$.

\begin{algorithm}[htb]
\caption{MCCFR w/ baseline}
\label{alg:bmccfr}
\begin{algorithmic}[1]
\Statex
\Function{MCCFR}{$\hist$}
    \Iif{$\hist \in \Terminals$}{\Return $\utility(\hist)$}
    \State $\strategy^\time(\hist, \cdot) \gets \Call{RegretMatching}{\Regret{\time-1}(\Iset(\hist), \cdot)}$
    \State $\avgstrat^\time(\hist, \cdot) \gets \frac{\time-1}{\time} \avgstrat^{\time - 1}(\hist, \cdot) + \frac{1}{\time} \strategy^\time$
    \State sample action $\act \sim \sampling{\time}(\hist, \cdot)$
    \State $\baseutil{(\hist\act)}{\strategy^\time}{\term^\time} \gets \Call{MCCFR}{(\hist\act)}$\label{line:recursion}
    \State $\baseutil{\hist, \act'}{\strategy^\time}{\term^\time} \gets \baseline^\time(\hist, \act') \qquad \forall \act' \neq \act$
    \State $\baseutil{\hist, \act}{\strategy^\time}{\term^\time} \gets \baseline^\time(\hist, \act) + \frac{1}{\sampling{\time}(\hist, \act)}\left(\baseutil{(\hist\act)}{\strategy^\time}{\term^\time} - \baseline^\time(\hist, \act)\right)$
    \State $\baseutil{\hist}{\strategy^\time}{\term^\time} \gets \sum_{\act'} \strategy^\time(\hist, \act') \baseutil{\hist, \act'}{\strategy^\time}{\term^\time}$
    \If {$\Toact(\hist) = 1$}
        \State $\regret{\time}(\Iset(\hist), \act) \gets \frac{\reach{\strategy^\time}_2(\hist)}{\reach{\sampling{\time}}(\hist)}\left(\baseutil{\hist, \cdot}{\strategy^\time}{\term^\time} - \baseutil{\hist}{\strategy^\time}{\term^\time}\right)$
    \ElsIf {$\Toact(\hist) = 2$}
        \State $\regret{\time}(\Iset(\hist), \act) \gets \frac{\reach{\strategy^\time}_1(\hist)}{\reach{\sampling{\time}}(\hist)}\left({-\baseutil{\hist, \cdot}{\strategy^\time}{\term^\time}} + \baseutil{\hist}{\strategy^\time}{\term^\time}\right)$
    \EndIf
    \State $\Regret{\time}(\Iset(\hist), \cdot) \gets \Regret{\time-1}(\Iset(\hist), \cdot) + \regret{\time}(\Iset(\hist), \cdot)$\label{line:regupdate}
    \State $\baseline^{\time+1}(\hist,\act) \gets \Call{UpdateBaseline}{\baseline^\time(\hist,\act), \baseutil{(\hist\act)}{\strategy^\time}{\term^\time}}$
    \State \Return $\baseutil{\hist}{\strategy^\time}{\term^\time}$\label{line:return}
\EndFunction
\end{algorithmic}
\end{algorithm}

\begin{algorithm}[htb]
\caption{MCCFR w/ predictive baseline}
\label{alg:predmccfr}
\begin{algorithmic}[1]
\Statex
\Function{MCCFR}{$\hist$}
    \Iif{$\hist \in \Terminals$}{\Return $\utility(\hist), \utility(\hist)$}
    \State $\strategy^\time(\hist, \cdot) \gets \Call{RegretMatching}{\Regret{\time-1}(\Iset(\hist), \cdot)}$
    \State $\avgstrat^\time(\hist, \cdot) \gets \frac{\time-1}{\time} \avgstrat^{\time - 1}(\hist, \cdot) + \frac{1}{\time} \strategy^\time$
    \State sample action $\act \sim \sampling{\time}(\hist, \cdot)$
    \State $\baseutil{(\hist\act)}{\strategy^\time}{\term^\time}, \baseutil{(\hist\act)}{\strategy^{\time+1}}{\term^\time} \gets \Call{MCCFR}{(\hist\act)}$
    \State $\baseutil{\hist, \act'}{\strategy^\time}{\term^\time} \gets \baseline^\time(\hist, \act') \qquad \forall \act' \neq \act$
    \State $\baseutil{\hist, \act}{\strategy^\time}{\term^\time} \gets \baseline^\time(\hist, \act) + \frac{1}{\sampling{\time}(\hist, \act)}\left(\baseutil{(\hist\act)}{\strategy^\time}{\term^\time} - \baseline^\time(\hist, \act)\right)$
    \State $\baseutil{\hist}{\strategy^\time}{\term^\time} \gets \sum_{\act'} \strategy^\time(\hist, \act') \baseutil{\hist, \act'}{\strategy^\time}{\term^\time}$
    \If {$\Toact(\hist) = 1$}
        \State $\regret{\time}(\Iset(\hist), \act) \gets \frac{\reach{\strategy^\time}_2(\hist)}{\reach{\sampling{\time}}(\hist)}\left(\baseutil{\hist, \cdot}{\strategy^\time}{\term^\time} - \baseutil{\hist}{\strategy^\time}{\term^\time}\right)$
    \ElsIf {$\Toact(\hist) = 2$}
        \State $\regret{\time}(\Iset(\hist), \act) \gets \frac{\reach{\strategy^\time}_1(\hist)}{\reach{\sampling{\time}}(\hist)}\left({-\baseutil{\hist, \cdot}{\strategy^\time}{\term^\time}} + \baseutil{\hist}{\strategy^\time}{\term^\time}\right)$
    \EndIf
    \State $\Regret{\time}(\Iset(\hist), \cdot) \gets \Regret{\time-1}(\Iset(\hist), \cdot) + \regret{\time}(\Iset(\hist), \cdot)$
    \State $\baseline^{\time+1}(\hist,\act) \gets \baseutil{(\hist\act)}{\strategy^{\time+1}}{\term^\time}$
    \State $\strategy^{\time+1}(\hist, \cdot) \gets \Call{RegretMatching}{\Regret{\time}(\Iset(\hist), \cdot)}$
    \State $\baseutil{\hist, \act'}{\strategy^{\time+1}}{\term^\time} \gets \baseline^\time(\hist, \act') \qquad \forall \act' \neq \act$
    \State $\baseutil{\hist, \act}{\strategy^{\time+1}}{\term^\time} \gets \baseline^{\time+1}(\hist, \act)$
    \State $\baseutil{\hist}{\strategy^{\time+1}}{\term^\time} \gets \sum_{\act'} \strategy^{\time+1}(\hist, \act') \baseutil{\hist, \act'}{\strategy^{\time+1}}{\term^\time}$
    \State \Return $\baseutil{\hist}{\strategy^\time}{\term^\time}, \baseutil{\hist}{\strategy^{\time+1}}{\term^\time}$
\EndFunction
\end{algorithmic}
\end{algorithm}

\section{Proof of Theorem~\ref{thm:unbiased}}

This proof is a simplified version of the proof of Lemma~5 in \citet{Schmid19}.

We directly analyze the expectation of the baseline-corrected utility:
\begin{align*}
&\Exp[\term^\time]{\baseutil{\hist, \act}{\strategy^\time}{\term^\time} \mid \term^\time \sqsupseteq \hist}\\
&= \Prob{(\hist\act) \sqsubseteq \term^\time \mid \hist \sqsubseteq \term^\time}\left(\frac{1}{\sampling{\time}(\hist, \act)}\left(\Exp[\term^\time]{\baseutil{(\hist\act)}{\strategy^\time}{\term^\time} \mid \term^\time \sqsupseteq (\hist\act)} - \baseline^\time(\hist, \act)\right) + \baseline^\time(\hist, \act)\right)\\
&\quad + \Prob{(\hist\act) \not\sqsubseteq \term^\time \mid \hist \sqsubseteq \term^\time}(\baseline^\time(\hist, \act))\\
&= \sampling{\time}(\hist, \act)\left(\frac{1}{\sampling{\time}(\hist, \act)}\left(\Exp[\term^\time]{\baseutil{(\hist\act)}{\strategy^\time}{\term^\time} \mid \term^\time \sqsupseteq (\hist\act)} - \baseline^\time(\hist, \act)\right) + \baseline^\time(\hist, \act)\right)\\
&\quad + (1 - \sampling{\time}(\hist, \act))(\baseline^\time(\hist, \act))\\
&= \Exp[\term^\time]{\baseutil{(\hist\act)}{\strategy^\time}{\term^\time} \mid \term^\time \sqsupseteq (\hist\act)}
\end{align*}

We now proceed by induction on the height of $(\hist\act)$ in the tree. If $(\hist\act)$ has height 0, then $(\hist\act) \in \Terminals$ and $\Exp[\term^\time]{\baseutil{\hist, \act}{\strategy^\time}{\term^\time} \mid \term^\time \sqsupseteq \hist} = \Exp[\term^\time]{\baseutil{(\hist\act)}{\strategy^\time}{\term^\time} \mid \term^\time \sqsupseteq (\hist\act)} = \utility((\hist\act))$ by definition.

For the inductive step, consider arbitrary $\hist, \act$ such that $(\hist\act)$ has height more than 0. We assume that $\Exp[\term^\time]{\baseutil{\hist', \act'}{\strategy^\time}{\term^\time} \mid \term^\time \sqsupseteq \hist'} = \exputil{(\hist'\act')}{\strategy^\time}$ for all $\hist', \act'$ such that $(\hist'\act')$ has smaller height than $(\hist\act)$. We then have
\begin{align*}
&\hspace{-10pt}\Exp[\term^\time]{\baseutil{\hist, \act}{\strategy^\time}{\term^\time} \mid \term^\time \sqsupseteq \hist}\\
&= \Exp[\term^\time]{\baseutil{(\hist\act)}{\strategy^\time}{\term^\time} \mid \term^\time \sqsupseteq (\hist\act)}\\
&= \sum_{\act' \in \Actions((\hist\act))} \strategy^\time((\hist\act), \act') \Exp[\term^\time]{\baseutil{(\hist\act), \act'}{\strategy^\time}{\term^\time} \mid \term^\time \sqsupseteq (\hist\act)}\\
&= \sum_{\act' \in \Actions((\hist\act))} \strategy^\time((\hist\act), \act') \exputil{(\hist\act\act')}{\strategy^\time} &&\text{by inductive hypothesis}\\
&= \exputil{(\hist\act)}{\strategy^\time} &&\text{by definition}
\end{align*}
We are able to apply the inductive hypothesis because $(\hist\act\act')$ is a suffix of $(\hist\act)$ and thus must have smaller height. The proof follows by induction. \qed

\section{Proof of Theorem \ref{thm:variance}}

This proof is similar to the proof of Lemma~3 in \citet{Schmid19}.

We begin by proving that the assumption of the theorem necessitates that the baseline-corrected values are the true expected values.

\begin{lem} Assume that we have a baseline that satisfies $\baseline^\time(\hist,\act) = \exputil{(\hist\act)}{\strategy^\time}$ for all $\hist \in \Histories$, $\act \in \Actions(\hist)$. Then for any $\hist, \act, \term^\time$,
\begin{equation}
\baseutil{\hist, \act}{\strategy^\time}{\term^\time} = \exputil{(\hist\act)}{\strategy^\time}
\end{equation}
\label{lem:trueval}
\end{lem}

\begin{proof}
As for Theorem~\ref{thm:unbiased}, we prove this by induction on the height of $(\hist\act)$. If $(\hist\act) \in \Terminals$, then by definition $\baseutil{(\hist\act)}{\strategy^\time}{\term^\time} = \utility((\hist\act)) = \exputil{(\hist\act)}{\strategy^\time}$. This then means
\begin{align*}
&\hspace{-10pt}\baseutil{\hist, \act}{\strategy^\time}{\term^\time}\\
&= \frac{\indic((\hist\act) \sqsubseteq \term^\time)}{\sampling{\time}(\hist, \act)} \left(\baseutil{(\hist\act)}{\strategy^\time}{\term^\time} - \baseline^\time(\hist, \act)\right) + \baseline^\time(\hist, \act)\\
&= \frac{\indic((\hist\act) \sqsubseteq \term^\time)}{\sampling{\time}(\hist, \act)} \left(\exputil{(\hist\act)}{\strategy^\time} - \baseline^\time(\hist, \act)\right) + \baseline^\time(\hist, \act)\\
&= \frac{\indic((\hist\act) \sqsubseteq \term^\time)}{\sampling{\time}(\hist, \act)} \left(\exputil{(\hist\act)}{\strategy^\time} - \exputil{(\hist\act)}{\strategy^\time}\right) + \exputil{(\hist\act)}{\strategy^\time} &&\text{by lemma assumption}\\
&= \exputil{(\hist\act)}{\strategy^\time}
\end{align*}

For the inductive step, we assume that $\baseutil{\hist', \act'}{\strategy^\time}{\term^\time} = \exputil{(\hist'\act')}{\strategy^\time}$ for any $(\hist'\act')$ with smaller height than $(\hist\act)$. Then
\begin{align*}
\baseutil{(\hist\act)}{\strategy^\time}{\term^\time} &= \sum_{\act' \in \Actions((\hist\act))} \strategy^\time((\hist\act), \act') \baseutil{(\hist\act), \act'}{\strategy^\time}{\term^\time}\\
&= \sum_{\act' \in \Actions((\hist\act))} \strategy^\time((\hist\act), \act') \exputil{(\hist\act\act')}{\strategy^\time} &&\text{by inductive hypothesis}\\
&= \exputil{(\hist\act)}{\strategy^\time}
\end{align*}
and this in turn gives us
\begin{align*}
&\hspace{-10pt}\baseutil{\hist, \act}{\strategy^\time}{\term^\time}\\
&= \frac{\indic((\hist\act) \sqsubseteq \term^\time)}{\sampling{\time}(\hist, \act)} \left(\exputil{(\hist\act)}{\strategy^\time} - \baseline^\time(\hist, \act)\right) + \baseline^\time(\hist, \act)\\
&= \frac{\indic((\hist\act) \sqsubseteq \term^\time)}{\sampling{\time}(\hist, \act)} \left(\exputil{(\hist\act)}{\strategy^\time} - \exputil{(\hist\act)}{\strategy^\time}\right) + \exputil{(\hist\act)}{\strategy^\time} &&\text{by lemma assumption}\\
&= \exputil{(\hist\act)}{\strategy^\time}
\end{align*}
The proof of the lemma follows by induction.
\end{proof}

To finish the proof of Theorem~\ref{thm:unbiased}, we only have to note that
\begin{align*}
\Var[\term^\time]{\baseutil{\hist, \act}{\strategy^\time}{\term^\time} | \term^\time \sqsupseteq \hist} &= \Var[\term^\time]{\exputil{(\hist\act)}{\strategy^\time} | \term^\time \sqsupseteq \hist} &&\text{by Lemma~\ref{lem:trueval}}\\
&= 0
\end{align*}
The last step follows because the expected utility is not a random variable. \qed

\section{Further variance analysis}

Theorem~\ref{thm:variance} shows that if the baseline function exactly predicts the true expected utility, then the baseline-corrected sampled values will have zero variance. However, it doesn't show what happens to the variance when the baseline function only approximates the expected utility. In this section, we show that the variance is a function of the differences between the baseline estimates and the true expected values.

\setcounter{thm}{3}
\begin{thm}
For any baseline function $\baseline^\time$ and any $\hist, \act$
\begin{equation*}
\Var[\term^\time]{\baseutil{\hist, \act}{\strategy^\time}{\term^\time} | \term^\time \sqsupseteq \hist} \leq \sum_{(\hist'\act') \sqsupseteq (\hist\act)} \frac{(\reach{\strategy^\time}((\hist\act), (\hist'\act')))^2}{\reach{\sampling{\time}}(\hist, (\hist'\act'))} \left(\exputil{(\hist'\act')}{\strategy^\time} - \baseline^\time(\hist', \act') \right)^2
\end{equation*}
\label{thm:fullvar}
\end{thm}

Theorem~\ref{thm:fullvar} shows that as the baseline estimates $\baseline^\time(\hist, \act)$ approach the expected values $\exputil{(\hist\act)}{\strategy^\time}$, the variance converges to zero. It also establishes a bound on how quickly it happens.

Before proving Theorem~\ref{thm:fullvar}, we first examine how the full (trajectory) variance can be decomposed into contributions from individual actions.

\begin{lem}
For any baseline function $\baseline^\time$ and any $\hist \in \Histories$
\begin{align*}
\Var[\term^\time]{\baseutil{\hist}{\strategy^\time}{\term^\time} | \term^\time \sqsupseteq \hist} = &\sum_{\act \in \Actions(\hist)} \frac{\left(\strategy^\time(\hist, \act)\right)^2}{\sampling{\time}(\hist, \act)} \Var[\term^\time]{\baseutil{(\hist\act)}{\strategy^\time}{\term^\time}}\\ &~+ \Var[\act]{\frac{\strategy^\time(\hist, \act)}{\sampling{\time}(\hist, \act)} \left(\exputil{(\hist\act)}{\strategy^\time} - \baseline^\time(\hist, \act) \right)}\\
\end{align*}
\label{lem:recvar}
\end{lem}

\begingroup
\allowdisplaybreaks

\begin{proof}
We use the law of total variance, conditioning on which $\act$ is sampled at $\hist$. This gives us
\begin{align}
&\Var[\term^\time]{\baseutil{\hist}{\strategy^\time}{\term^\time}  | \term^\time \sqsupseteq \hist}\nonumber\\&\quad = \Exp[\act]{\Var[\term^\time]{\baseutil{\hist}{\strategy^\time}{\term^\time} \middle| \term^\time \sqsupseteq (\hist\act)}} + \Var[\act]{\Exp[\term^\time]{\baseutil{\hist}{\strategy^\time}{\term^\time} \middle| \term^\time \sqsupseteq (\hist\act)}}\label{eq:ltv}
\end{align}
We analyze each of these terms separately.

First, to analyze the left summand in (\ref{eq:ltv}), we note that if $\hist\act \sqsubset \term^\time$ then by the recursive definition of baseline-corrected values
\begin{align*}
\baseutil{\hist}{\strategy^\time}{\term^\time} = \frac{\strategy^\time(\hist, \act)}{\sampling{\time}(\hist, \act)} \baseutil{(\hist\act)}{\strategy^\time}{\term^\time} - \frac{\strategy^\time(\hist, \act)}{\sampling{\time}(\hist, \act)} \baseline^\time(\hist, \act) +  \sum_{\act' \in \Actions(\hist)} \strategy^\time(\hist, \act') \baseline^\time(\hist, \act')
\end{align*}
Only the first term depends on the sampled trajectory $\term^\time$, and thus
\begin{align}
\Exp[\act]{\Var[\term^\time]{\baseutil{\hist}{\strategy^\time}{\term^\time} \middle| \term^\time \sqsupseteq (\hist\act)}} &= \Exp[\act]{\Var[\term^\time]{\frac{\strategy^\time(\hist, \act)}{\sampling{\time}(\hist, \act)} \baseutil{(\hist\act)}{\strategy^\time}{\term^\time} \middle| \term^\time \sqsupseteq (\hist\act)}}\nonumber\\
&= \Exp[\act]{\left(\frac{\strategy^\time(\hist, \act)}{\sampling{\time}(\hist, \act)}\right)^2\Var[\term^\time]{\baseutil{(\hist\act)}{\strategy^\time}{\term^\time} \middle| \term^\time \sqsupseteq (\hist\act)}}\nonumber\\
&=\sum_{\act \in \Actions(\hist)} \frac{\left(\strategy^\time(\hist, \act)\right)^2}{\sampling{\time}(\hist, \act)} \Var[\term^\time]{\baseutil{(\hist\act)}{\strategy^\time}{\term^\time} \middle| \term^\time \sqsupseteq (\hist\act)}\label{eq:lem2eq2}
\end{align}

Next, we analyze the inner expectation of the right summand of (\ref{eq:ltv})
\begin{align*}
&\hspace{-30pt}\Exp[\term^\time]{\baseutil{\hist}{\strategy^\time}{\term^\time} \middle| \term^\time \sqsupseteq (\hist\act)}\\
&=  \sum_{\act'} \strategy^\time(\hist, \act') \baseline^\time(\hist, \act') + \frac{\strategy^\time(\hist, \act)}{\sampling{\time}(\hist, \act)} \left(\Exp[\term^\time]{\baseutil{(\hist\act)}{\strategy^\time}{\term^\time}} - \baseline^\time(\hist, \act) \right)\\
&= \sum_{\act'} \strategy^\time(\hist, \act') \baseline^\time(\hist, \act') + \frac{\strategy^\time(\hist, \act)}{\sampling{\time}(\hist, \act)} \left(\exputil{(\hist\act)}{\strategy^\time} - \baseline^\time(\hist, \act) \right)
\end{align*}
The first term here doesn't depend on the sampled $\act$, giving us
\begin{equation}
\Var[\act]{\Exp[\term^\time]{\baseutil{\hist}{\strategy^\time}{\term^\time} \middle| (\hist\act) \sqsubseteq \term^\time}} = \Var[\act]{\frac{\strategy^\time(\hist, \act)}{\sampling{\time}(\hist, \act)} \left(\exputil{(\hist\act)}{\strategy^\time} - \baseline^\time(\hist, \act) \right)}\label{eq:lem2eq3}
\end{equation}

Combining (\ref{eq:ltv}), (\ref{eq:lem2eq2}), and (\ref{eq:lem2eq3}) completes the proof.
\end{proof}

Lemma~\ref{lem:recvar} decomposes the variance into a part from the immediately sampled action, and a part from the remainder of the sampled trajectory. We extend this to completely decompose the trajectory variance.

\begin{lem}
For any baseline function $\baseline^\time$ and any $\hist, \act$
\begin{equation*}
\Var[\term^\time]{\baseutil{\hist}{\strategy^\time}{\term^\time} | \term^\time \sqsupseteq \hist} = \sum_{\hist' \sqsupseteq \hist} \frac{(\reach{\strategy^\time}(\hist, \hist'))^2}{\reach{\sampling{\time}}(\hist, \hist')} \Var[\act']{\frac{\strategy^\time(\hist', \act')}{\sampling{\time}(\hist', \act')} \left(\exputil{(\hist'\act')}{\strategy^\time} - \baseline^\time(\hist', \act') \right)}
\end{equation*}
\label{lem:fullvar}
\end{lem}
\begin{proof}
We proceed by induction on the height of $\hist$ in the tree. If $\hist$ has height 0, then $\Actions(\hist) = \emptyset$, and $\Var[\term^\time]{\baseutil{\hist}{\strategy^\time}{\term^\time} | \term^\time \sqsupseteq \hist} = 0$. Otherwise, we begin from Lemma~\ref{lem:recvar} and apply the inductive hypothesis for $\hist'$ with height less than that of $\hist$. This gives
\begin{align*}
&\hspace{-20pt}\Var[\term^\time]{\baseutil{\hist}{\strategy^\time}{\term^\time} | \term^\time \sqsupseteq \hist}\\ 
&= \sum_{\act \in \Actions(\hist)} \frac{\left(\strategy^\time(\hist, \act)\right)^2}{\sampling{\time}(\hist, \act)} \Var[\term^\time]{\baseutil{(\hist\act)}{\strategy^\time}{\term^\time}} + \Var[\act]{\frac{\strategy^\time(\hist, \act)}{\sampling{\time}(\hist, \act)} \left(\exputil{(\hist\act)}{\strategy^\time} - \baseline^\time(\hist, \act) \right)}\\
&= \sum_{\act \in \Actions(\hist)} \frac{\left(\strategy^\time(\hist, \act)\right)^2}{\sampling{\time}(\hist, \act)} \sum_{\hist' \sqsupseteq (\hist\act)} \frac{(\reach{\strategy^\time}((\hist\act), \hist'))^2}{\reach{\sampling{\time}}((\hist\act), \hist')} \Var[\act']{\frac{\strategy^\time(\hist', \act')}{\sampling{\time}(\hist', \act')} \left(\exputil{(\hist'\act')}{\strategy^\time} - \baseline^\time(\hist', \act') \right)}\\
&\qquad + \Var[\act]{\frac{\strategy^\time(\hist, \act)}{\sampling{\time}(\hist, \act)} \left(\exputil{(\hist\act)}{\strategy^\time} - \baseline^\time(\hist, \act) \right)}\\
&= \sum_{\act \in \Actions(\hist)} \sum_{\hist' \sqsupseteq (\hist\act)} \frac{(\reach{\strategy^\time}(\hist, \hist'))^2}{\reach{\sampling{\time}}(\hist, \hist')} \Var[\act']{\frac{\strategy^\time(\hist', \act')}{\sampling{\time}(\hist', \act')} \left(\exputil{(\hist'\act')}{\strategy^\time} - \baseline^\time(\hist', \act') \right)}\\
&\qquad + \Var[\act]{\frac{\strategy^\time(\hist, \act)}{\sampling{\time}(\hist, \act)} \left(\exputil{(\hist\act)}{\strategy^\time} - \baseline^\time(\hist, \act) \right)}\\
&= \sum_{\hist' \sqsupseteq \hist} \frac{(\reach{\strategy^\time}(\hist, \hist'))^2}{\reach{\sampling{\time}}(\hist, \hist')} \Var[\act']{\frac{\strategy^\time(\hist', \act')}{\sampling{\time}(\hist', \act')} \left(\exputil{(\hist'\act')}{\strategy^\time} - \baseline^\time(\hist', \act') \right)}
\end{align*}
The lemma follows by induction.
\end{proof}

\begin{proof}[Proof of Theorem~\ref{thm:fullvar}]
Starting from Lemma~\ref{lem:fullvar}, we first bound the variance of history values
\begin{align}
&\hspace{-20pt}\Var[\term^\time]{\baseutil{\hist}{\strategy^\time}{\term^\time} | \term^\time \sqsupseteq \hist}\nonumber\\
&= \sum_{\hist' \sqsupseteq \hist} \frac{(\reach{\strategy^\time}(\hist, \hist'))^2}{\reach{\sampling{\time}}(\hist, \hist')} \Var[\act']{\frac{\strategy^\time(\hist', \act')}{\sampling{\time}(\hist', \act')} \left(\exputil{(\hist'\act')}{\strategy^\time} - \baseline^\time(\hist', \act') \right)}\nonumber\\
&\leq \sum_{\hist' \sqsupseteq \hist} \frac{(\reach{\strategy^\time}(\hist, \hist'))^2}{\reach{\sampling{\time}}(\hist, \hist')} \Exp[\act']{\left(\frac{\strategy^\time(\hist', \act')}{\sampling{\time}(\hist', \act')} \left(\exputil{(\hist'\act')}{\strategy^\time} - \baseline^\time(\hist', \act') \right)\right)^2}\nonumber\\
&= \sum_{\hist' \sqsupseteq \hist} \frac{(\reach{\strategy^\time}(\hist, \hist'))^2}{\reach{\sampling{\time}}(\hist, \hist')} \sum_{\act' \in \Actions(\hist')} \frac{\left(\strategy^\time(\hist', \act')\right)^2}{\sampling{\time}(\hist', \act')} \left(\exputil{(\hist'\act')}{\strategy^\time} - \baseline^\time(\hist', \act') \right)^2\nonumber\\
&= \sum_{\substack{\hist' \sqsupseteq \hist\\ \act' \in \Actions(\hist')}} \frac{(\reach{\strategy^\time}(\hist, (\hist'\act')))^2}{\reach{\sampling{\time}}(\hist, (\hist'\act'))} \left(\exputil{(\hist'\act')}{\strategy^\time} - \baseline^\time(\hist', \act') \right)^2\label{eq:thm4eq1}
\end{align}

We then reformulate the variance of the history action value $\baseutil{\hist, \act}{\strategy^\time}{\term^\time}$ in terms of the variance of the succeeding history value $\baseutil{(\hist\act)}{\strategy^\time}{\term^\time}$. To do this, we apply the law of total variance conditioning on the random variable $\indic((\hist\act) \sqsubseteq \term^\time)$ which indicates whether $\act$ is sampled at $\hist$.
\begin{align}
&\hspace{-20pt}\Var[\term^\time]{\baseutil{\hist, \act}{\strategy^\time}{\term^\time} \middle| \term^\time \sqsupseteq \hist}\nonumber\\
&= \Var[\term^\time]{\frac{\indic((\hist\act) \sqsubseteq \term^\time)}{\sampling{\time}(\hist, \act)} \left(\baseutil{(\hist\act)}{\strategy^\time}{\term^\time} - \baseline^\time(\hist, \act)\right) + \baseline^\time(\hist, \act) \middle| \term^\time \sqsupseteq \hist}\nonumber\\
&= \Var[\term^\time]{\frac{\indic((\hist\act) \sqsubseteq \term^\time)}{\sampling{\time}(\hist, \act)} \left(\baseutil{(\hist\act)}{\strategy^\time}{\term^\time} - \baseline^\time(\hist, \act)\right) \middle| \term^\time \sqsupseteq \hist}\nonumber\\
&= \Exp{\Var[\term^\time]{\frac{\indic((\hist\act) \sqsubseteq \term^\time)}{\sampling{\time}(\hist, \act)} \left(\baseutil{(\hist\act)}{\strategy^\time}{\term^\time} - \baseline^\time(\hist, \act)\right) \middle| \indic((\hist\act) \sqsubseteq \term^\time)}}\nonumber\\
&\qquad + \Var{\Exp[\term^\time]{\frac{\indic((\hist\act) \sqsubseteq \term^\time)}{\sampling{\time}(\hist, \act)} \left(\baseutil{(\hist\act)}{\strategy^\time}{\term^\time} - \baseline^\time(\hist, \act)\right) \middle| \indic((\hist\act) \sqsubseteq \term^\time)}}\nonumber\\
&= \Exp{\frac{\indic((\hist\act) \sqsubseteq \term^\time)}{(\sampling{\time}(\hist, \act))^2}\Var[\term^\time]{\baseutil{(\hist\act)}{\strategy^\time}{\term^\time} \middle| \indic((\hist\act) \sqsubseteq \term^\time)}}\nonumber\\
&\qquad + \Var{\frac{\indic((\hist\act) \sqsubseteq \term^\time)}{\sampling{\time}(\hist, \act)} \left(\exputil{(\hist\act)}{\strategy^\time} - \baseline^\time(\hist, \act)\right)}\nonumber\\
&= \frac{1}{\sampling{\time}(\hist, \act)}\Var[\term^\time]{\baseutil{(\hist\act)}{\strategy^\time}{\term^\time} \middle| \term^\time \sqsupseteq (\hist\act)}\nonumber\\
&\qquad + \frac{1}{(\sampling{\time}(\hist, \act))^2}\left(\exputil{(\hist\act)}{\strategy^\time} - \baseline^\time(\hist, \act)\right)^2\Var{\indic((\hist\act) \sqsubseteq \term^\time)}\nonumber\\
&= \frac{1}{\sampling{\time}(\hist, \act)}\Var[\term^\time]{\baseutil{(\hist\act)}{\strategy^\time}{\term^\time} \middle| \term^\time \sqsupseteq (\hist\act)} + \frac{1 - \sampling{\time}(\hist, \act)}{\sampling{\time}(\hist, \act)}\left(\exputil{(\hist\act)}{\strategy^\time} - \baseline^\time(\hist, \act)\right)^2\nonumber\\
&\leq \frac{1}{\sampling{\time}(\hist, \act)}\left(\Var[\term^\time]{\baseutil{(\hist\act)}{\strategy^\time}{\term^\time} \middle| \term^\time \sqsupseteq (\hist\act)} + \left(\exputil{(\hist\act)}{\strategy^\time} - \baseline^\time(\hist, \act)\right)^2\right)\label{eq:thm4eq2}
\end{align}

Combining (\ref{eq:thm4eq1}) and (\ref{eq:thm4eq2}), we get
\begin{align*}
&\hspace{-20pt}\Var[\term^\time]{\baseutil{\hist, \act}{\strategy^\time}{\term^\time} \middle| \term^\time \sqsupseteq \hist}\\
&\leq \frac{1}{\sampling{\time}(\hist, \act)}\biggl(\sum_{\substack{\hist' \sqsupseteq (\hist\act)\\ \act' \in \Actions(\hist')}} \frac{(\reach{\strategy^\time}((\hist\act), (\hist'\act')))^2}{\reach{\sampling{\time}}((\hist\act), (\hist'\act'))} \left(\exputil{(\hist'\act')}{\strategy^\time} - \baseline^\time(\hist', \act') \right)^2\\
&\hspace{2cm} + \left(\exputil{(\hist\act)}{\strategy^\time} - \baseline^\time(\hist, \act)\right)^2\biggr)\\
&= \frac{1}{\sampling{\time}(\hist, \act)}\sum_{(\hist'\act') \sqsupseteq (\hist\act)} \frac{(\reach{\strategy^\time}((\hist\act), (\hist'\act')))^2}{\reach{\sampling{\time}}((\hist\act), (\hist'\act'))} \left(\exputil{(\hist'\act')}{\strategy^\time} - \baseline^\time(\hist', \act') \right)^2\\
&= \sum_{(\hist'\act') \sqsupseteq (\hist\act)} \frac{(\reach{\strategy^\time}((\hist\act), (\hist'\act')))^2}{\reach{\sampling{\time}}(\hist, (\hist'\act'))} \left(\exputil{(\hist'\act')}{\strategy^\time} - \baseline^\time(\hist', \act') \right)^2
\end{align*}

\end{proof}

\endgroup
\section{Public trees}
\label{sec:pubtrees}

There are multiple sources of variance when computing the regret at an information set in MCCFR. One form of variance comes from sampling actions (and recursively, trajectories) from the information set, rather than walking the full subtree. A second form of variance comes from sampling only one of the histories in the information set itself. Our baseline framework reduces the first kind of variance, but does not take the second form of variance into account.

One approach to combating this single-history variance could be to extend the use of the baseline; analogous to how we created a control variate from using $\baseline^\time(\hist, \act)$ to evaluate unsampled actions $\act$, we could also create a control variate that uses $\baseline^\time(\hist', \act)$ to evaluate all unsampled $\hist' \in \Iset(\hist)$. However, this requires evaluating alternate histories along every step of the sampled trajectory, meaning that a single iteration of MCCFR goes from complexity $\BigO{d|\Actions_{\text{max}}|}$ to $\BigO{d|\Actions_{\text{max}}||\Iset_{\text{max}}}|$.

A second approach, and the one we present in this section, is to change the sampling method used. Rather than using a baseline to consider each alternate history in the information set, we directly evaluate all such histories. Intuitively, this can be done by only sampling actions that are publicly observable, and walking all actions that change the game's hidden state. This approach was used by \citet{Schmid19}, but was never formalized. We formalize the algorithm here, after presenting some additional assumptions and definitions.

We assume that the EFG is \emph{timeable} \citep{Jakobsen16}, which informally means that no player can gain additional information by tracking how much time elapses while they are not acting. Formally, this means that we can assign a value $\timefunc{\hist}$ to every $\hist \in \Histories$ such that $\timefunc{\hist} = \timefunc{\hist'}$ for any $\hist' \in \Iset_{\Toact(\hist)}(\hist)$, and $\timefunc{\hist} < \timefunc{\hist'}$ for any $\hist' \sqsupseteq \hist$. Every game played by humans must be timeable, or else the human could distinguish histories in the same information set by tracking elapsed time. If a game is timeable, players always observe the timing when they are acting, so there must be some strategically identical game where they observe the timing even when not acting. Thus we will assume that our games satisfy this requirement.

We now introduce the concept of a \emph{public state} \citep{Johanson11}, which groups histories based on information available to all players, or informally, based on whether they distinguishable to an outside observer. Formally, a public state is a set of histories that is (minimally) closed under the information set relation for all players. Let $\Pubstates$ be the set of public states (which partitions $\Histories$), and $\Pubstate(\hist) \in \Pubstates$ be the public state that $\hist$ belongs to. By assumption that all players observe the game's timing, necessarily $\timefunc{\hist} = \timefunc{\hist'}$ if $\Pubstate(\hist) = \Pubstate(\hist')$. In turn, this means that if $\hist \sqsubseteq \hist'$, then $\Pubstate(\hist) \neq \Pubstate(\hist')$. We also assume for simplicity that if $\Pubstate(\hist) = \Pubstate(\hist')$, then $\Toact(\hist) = \Toact(\hist')$. If necessary, this can be made true for any timeable game by splitting information sets and adding dummy actions, without strategically changing the game.

We define $\Transitions{\Pubstate}$ to be the set of successor public states to $\Pubstates$: $\Pubstate' \in \Transitions{\Pubstate}$ if there is some $\hist \in \Pubstate$, $\act \in \Actions(\hist)$, and $\hist' \in \Pubstate'$ such that $(\hist\act) = \hist'$. The successor relation defines the edges of a \emph{public tree}, where the public states are nodes. It should be noted that more than one action can lead to the same successor public state when some player doesn't observe the action, and that one action can lead to more than one successor public state if some previously private information becomes public.

In the statement of Theorem~\ref{thm:zerovar}, we used $\text{samp}^\time(\hist)$ to notate whether $\hist$ was sampled on iteration $\time$. With the notation introduced here, we can formalize this by defining $\text{samp}^\time(\hist)$ to occur if and only if $\hist' \sqsubseteq \term^\time$ for some $\hist' \in \Pubstate(\hist)$ and some $\term^\time \in \Terminals^\time$. For clarity, we thus symbolize this relation as $\Pubstate(\hist) \sqsubseteq \Terminals^\time$.

\subsection{Public Outcome Sampling}

We now define our MCCFR variant, which we call \emph{Public Outcome Sampling (POS)}. Instead of walking trajectories through the EFG tree by sampling actions, POS walks trajectories through the public tree by sampling successor public states.

For public state $\Pubstate$, let $\my{\Isets}(\Pubstate) \subseteq \my{\Isets}$ be the collection of player $\player$ information sets contained within $\Pubstate$. While walking down the tree, POS keeps track of reach probabilities $\my{\reach{\strategy^\time}}(\my{\Iset})$ for each $\my{\Iset} \in \my{\Isets}(\Pubstate)$ and each player $\player$ at public state $\Pubstate$. To recurse, it samples some successor $\Pubstate' \in \Transitions{\Pubstate}$ using a probability distribution $\sampling{\time}(\Pubstate) \in \Delta_{\Transitions{\Pubstate}}$. It updates the reach probabilities to $\my{\reach{\strategy^\time}}(\my{\Iset})$ for each $\my{\Iset} \in \my{\Isets}(\Pubstate')$, using the current strategy $\strategy^\time$. Ultimately, the recursion reaches a public state which only contains terminal nodes (as the end of the game is publicly observable). This public state, which defines the sampled trajectory in the public tree, we label $\Terminals^\time$. The terminal histories are evaluated as $\utility(\term)$ for each $\term \in \Terminals^\time$.

Walking back up the tree, at each recursion step we pass back the utilities $\baseutil{\hist'}{\strategy^\time}{\Terminals^\time}$ for each $\hist' \in \Pubstate'$. From these, we apply a baseline and recursively calculate utilities as
\begin{align}
\baseutil{\hist, \act}{\strategy^\time}{\Terminals^\time} &= \frac{\indic(\Pubstate((\hist\act)) \sqsubseteq \Terminals^\time)}{\sampling{\time}(\Pubstate(\hist), \Pubstate((\hist\act))}\left(\baseutil{(\hist\act)}{\strategy^\time}{\Terminals^\time} - \baseline^\time(\hist, \act)\right) + \baseline^\time(\hist, \act)\\
\baseutil{\hist}{\strategy^\time}{\Terminals^\time} &= \sum_{\act \in \Actions(\hist)} \strategy^\time(\hist, \act) \baseutil{\hist, \act}{\strategy^\time}{\Terminals^\time}
\end{align}
for each $\hist \in \Pubstate$ and $\act \in \Actions(\hist)$. We then use these values to calculate regrets $\regret{\time}(\Iset, \act)$ for each $\Iset \in \my{\Isets}(\Pubstate)$ and update the saved regrets.

Algorithm~\ref{alg:posmccfr} gives pseudocode for MCCFR with POS. Algorithm~\ref{alg:predposmccfr} gives pseudocode for a version using the predictive baseline.

\begin{algorithm}[htb]
\caption{MCCFR w/ POS and baseline}
\label{alg:posmccfr}
\begin{algorithmic}[1]
\Statex
\Function{POS-MCCFR}{$\Pubstate$}
    \Iif{$\Pubstate \subseteq \Terminals$}{\Return $\{\utility(\hist) \mid \forall \hist \in \Pubstate\}$}
    \For {$\Iset \in \Isets_{\Toact(\Pubstate)}(\Pubstate)$}
        \State $\strategy^\time(\Iset, \cdot) \gets \Call{RegretMatching}{\Regret{\time-1}(\Iset, \cdot)}$
        \State $\avgstrat^\time(\Iset, \cdot) \gets \frac{\time-1}{\time} \avgstrat^{\time - 1}(\Iset, \cdot) + \frac{1}{\time} \strategy^\time$
    \EndFor
    \State sample successor $\Pubstate' \sim \sampling{\time}(\Pubstate, \cdot)$
    \State $\{\baseutil{\hist'}{\strategy^\time}{\Terminals^\time} \mid \forall \hist' \in \Pubstate'\} \gets \Call{POS-MCCFR}{\Pubstate'}$
    \For {$\Iset \in \Isets_{\Toact(\Pubstate)}$}
        \For {$\hist \in \Iset$}
            \For {$\act \in \Actions(\hist)$}
                \If {$(\hist\act) \in \Pubstate'$}
                    \State $\baseutil{\hist, \act}{\strategy^\time}{\Terminals^\time} \gets \baseline^\time(\hist, \act) + \frac{1}{\sampling{\time}(\Pubstate, \Pubstate')}\left(\baseutil{(\hist\act)}{\strategy^\time}{\Terminals^\time} - \baseline^\time(\hist, \act)\right)$
                    \State $\baseline^{\time+1}(\hist,\act) \gets \Call{UpdateBaseline}{\baseline^\time(\hist,\act), \baseutil{(\hist\act)}{\strategy^\time}{\Terminals^\time}}$
                \Else
                    \State $\baseutil{\hist, \act}{\strategy^\time}{\Terminals^\time} \gets \baseline^\time(\hist, \act)$
                \EndIf
            \EndFor
            \State $\baseutil{\hist}{\strategy^\time}{\Terminals^\time} \gets \sum_{\act'} \strategy^\time(\hist, \act') \baseutil{\hist, \act'}{\strategy^\time}{\Terminals^\time}$
        \EndFor
        \If {$\Toact(\hist) = 1$}
            \State $\regret{\time}(\Iset, \act) \gets \frac{1}{\reach{\sampling{\time}}(\Pubstate)}\sum_{\hist \in \Iset}\reach{\strategy^\time}_2(\hist)\left(\baseutil{\hist, \cdot}{\strategy^\time}{\Terminals^\time} - \baseutil{\hist}{\strategy^\time}{\Terminals^\time}\right)$
        \ElsIf {$\Toact(\hist) = 2$}
            \State $\regret{\time}(\Iset, \act) \gets \frac{1}{\reach{\sampling{\time}}(\Pubstate)}\sum_{\hist \in \Iset}\reach{\strategy^\time}_1(\hist)\left({-\baseutil{\hist, \cdot}{\strategy^\time}{\Terminals^\time}} + \baseutil{\hist}{\strategy^\time}{\Terminals^\time}\right)$
        \EndIf
        \State $\Regret{\time}(\Iset, \cdot) \gets \Regret{\time-1}(\Iset, \cdot) + \regret{\time}(\Iset, \cdot)$
    \EndFor
    \State \Return $\{\baseutil{\hist}{\strategy^\time}{\Terminals^\time} \mid \forall \hist \in \Pubstate\}$
\EndFunction
\end{algorithmic}
\end{algorithm}

\begin{algorithm}[htb]
\caption{MCCFR w/ POS and predictive baseline}
\label{alg:predposmccfr}
\begin{algorithmic}[1]
\Statex
\Function{POS-MCCFR}{$\Pubstate$}
    \Iif{$\Pubstate \subseteq \Terminals$}{\Return $\{\utility(\hist), \utility(\hist) \mid \forall \hist \in \Pubstate\}$}
    \For {$\Iset \in \Isets_{\Toact(\Pubstate)}(\Pubstate)$}
        \State $\strategy^\time(\Iset, \cdot) \gets \Call{RegretMatching}{\Regret{\time-1}(\Iset, \cdot)}$
        \State $\avgstrat^\time(\Iset, \cdot) \gets \frac{\time-1}{\time} \avgstrat^{\time - 1}(\Iset, \cdot) + \frac{1}{\time} \strategy^\time$
    \EndFor
    \State sample successor $\Pubstate' \sim \sampling{\time}(\Pubstate, \cdot)$
    \State $\{\baseutil{\hist'}{\strategy^\time}{\Terminals^\time}, \baseutil{\hist'}{\strategy^{\time+1}}{\Terminals^\time} \mid \forall \hist' \in \Pubstate'\} \gets \Call{POS-MCCFR}{\Pubstate'}$
    \For {$\Iset \in \Isets_{\Toact(\Pubstate)}$}
        \For {$\hist \in \Iset$}
            \For {$\act \in \Actions(\hist)$}
                \If {$(\hist\act) \in \Pubstate'$}
                    \State $\baseutil{\hist, \act}{\strategy^\time}{\Terminals^\time} \gets \baseline^\time(\hist, \act) + \frac{1}{\sampling{\time}(\Pubstate, \Pubstate')}\left(\baseutil{(\hist\act)}{\strategy^\time}{\Terminals^\time} - \baseline^\time(\hist, \act)\right)$
                    \State $\baseline^{\time+1}(\hist,\act) \gets \baseutil{(\hist\act)}{\strategy^{\time+1}}{\Terminals^\time}$
                    \State $\baseutil{\hist, \act}{\strategy^{\time+1}}{\Terminals^\time} \gets \baseline^{\time+1}(\hist,\act)$
                \Else
                    \State $\baseutil{\hist, \act}{\strategy^\time}{\Terminals^\time} \gets \baseline^\time(\hist, \act)$
                    \State $\baseutil{\hist, \act}{\strategy^{\time+1}}{\Terminals^\time} \gets \baseline^\time(\hist, \act)$
                \EndIf
            \EndFor
            \State $\baseutil{\hist}{\strategy^\time}{\Terminals^\time} \gets \sum_{\act'} \strategy^\time(\hist, \act') \baseutil{\hist, \act'}{\strategy^\time}{\Terminals^\time}$
        \EndFor
        \If {$\Toact(\hist) = 1$}
            \State $\regret{\time}(\Iset, \act) \gets \frac{1}{\reach{\sampling{\time}}(\Pubstate)}\sum_{\hist \in \Iset}\reach{\strategy^\time}_2(\hist)\left(\baseutil{\hist, \cdot}{\strategy^\time}{\Terminals^\time} - \baseutil{\hist}{\strategy^\time}{\Terminals^\time}\right)$
        \ElsIf {$\Toact(\hist) = 2$}
            \State $\regret{\time}(\Iset, \act) \gets \frac{1}{\reach{\sampling{\time}}(\Pubstate)}\sum_{\hist \in \Iset}\reach{\strategy^\time}_1(\hist)\left({-\baseutil{\hist, \cdot}{\strategy^\time}{\Terminals^\time}} + \baseutil{\hist}{\strategy^\time}{\Terminals^\time}\right)$
        \EndIf
        \State $\Regret{\time}(\Iset, \cdot) \gets \Regret{\time-1}(\Iset, \cdot) + \regret{\time}(\Iset, \cdot)$
        \State $\strategy^{\time+1}(\Iset, \cdot) \gets \Call{RegretMatching}{\Regret{\time}(\Iset, \cdot)}$
        \For {$\hist \in \Iset$}
            \State $\baseutil{\hist}{\strategy^{\time+1}}{\Terminals^\time} \gets \sum_{\act'} \strategy^{\time+1}(\hist, \act') \baseutil{\hist, \act'}{\strategy^{\time+1}}{\Terminals^\time}$
        \EndFor
    \EndFor
    \State \Return $\{\baseutil{\hist}{\strategy^\time}{\Terminals^\time}, \baseutil{\hist}{\strategy^{\time+1}}{\Terminals^\time} \mid \forall \hist \in \Pubstate\}$
\EndFunction
\end{algorithmic}
\end{algorithm}

Updating a public state $\Pubstate$ with this algorithm requires walking through all of the possible histories in the public state, as well as all of the actions possible at each history, giving a complexity $\BigO{|\Pubstate||\Actions_{\text{max}}|}$, or equivalently $\BigO{|\my{\Isets}(\Pubstate)||\opp{\Isets}(\Pubstate)||\Actions_{\text{max}}|}$. However, the computations for each information set $\Iset \in \my{\Isets}(\Pubstate)$ with acting player $\player$ can be done completely independently, allowing for easy parallelization (e.g. on a GPU) to achieve complexity $\BigO{|\opp{\Isets}(\Pubstate)||\Actions_{\text{max}}|}$. This approach was taken with the non-sampling algorithm used in DeepStack \citep{Moravcik17}.

\section{Proof of Theorem~\ref{thm:zerovar}}
\label{sec:zerovar}

We introduce the following definition that tracks sampled values of terminal histories:
\begin{equation}
\partutil{\time}{\term} = \begin{cases} \utility(\term) \qquad &\text{if $\term \in \Terminals^\tau$ for any $\tau < \time$}\\ 0 &\text{otherwise} \end{cases}
\end{equation}

We begin by showing that the if our baseline function maintains is a weighted sum of values $\partutil{\time}{\term}$, then our predictive baseline-corrected value estimates will be as well.

\begin{lem}
Let $\Pubstate((\hist\act)) \sqsubseteq \Terminals^\time$, and assume that for all $(\hist'\act') \sqsupseteq (\hist\act)$, the predictive baseline function satisfies $\baseline^\time(\hist', \act') = \sum_{\term \in \Terminals[(\hist'\act')]} \reach{\strategy^\time}((\hist'\act'), \term) \partutil{\time}{\term}$. Then the baseline-corrected utility satisfies
\begin{equation}
\baseutil{\hist, \act}{\strategy^{\time+1}}{\Terminals^\time} = \sum_{\term \in \Terminals[(\hist\act)]} \reach{\strategy^{\time+1}}((\hist\act), \term) \partutil{\time+1}{\term}
\end{equation}
\label{lem:predupdate}
\end{lem}

\begin{proof}
We prove this by induction on the height of $(\hist\act)$ in the tree. Our base case is that $(\hist\act) = \term^\time$ for some $\term^\time \in \Terminals^\time$, in which case
\begin{align*}
&\hspace{-15pt}\baseutil{\hist, \act}{\strategy^{\time+1}}{\Terminals^\time}\\
&= \frac{1}{\sampling{\time}(\Pubstate(\hist),\Pubstate((\hist\act))}(\utility(\term^\time) - \baseline^{\time+1}(\hist, \act)) + \baseline^{\time+1}(\hist, \act)\\
&= \frac{1}{\sampling{\time}(\Pubstate(\hist),\Pubstate((\hist\act))}(\utility(\term^\time) - \utility(\term^\time)) + \utility(\term^\time) &&\text{by predictive baseline definition}\\
&= \utility(\term^\time)\\
&= \partutil{\time+1}{(\hist\act)} &&\text{as $(\hist\act)=\term^\time \in \Terminals^\time$}
\end{align*}

For the inductive step, we consider some $(\hist\act)$ and assume the hypothesis holds for all $(\hist'\act')$ with smaller height.
\begin{align*}
&\baseutil{\hist, \act}{\strategy^{\time+1}}{\Terminals^\time}\\
&= \frac{1}{\sampling{\time}(\Pubstate(\hist),\Pubstate((\hist\act))}\left(\baseutil{(\hist\act)}{\strategy^{\time+1}}{\Terminals^\time} - \baseline^{\time+1}(\hist, \act)\right) + \baseline^{\time+1}(\hist, \act)\\
&= \baseutil{(\hist\act)}{\strategy^{\time+1}}{\Terminals^\time} &&\hspace{-19pt}\text{by predictive baseline definition}\\
&= \sum_{\act' \in \Actions((\hist\act))} \strategy^{\time+1}((\hist\act), \act') \baseutil{(\hist\act), \act'}{\strategy^{\time+1}}{\Terminals^\time}\\
&= \sum_{\act' \in \Actions((\hist\act))} \strategy^{\time+1}((\hist\act), \act') \sum_{\term \in \Terminals[(\hist\act\act')]} \reach{\strategy^{\time+1}}((\hist\act\act'), \term) \partutil{\time+1}{\term} &&\text{by inductive hypothesis}\\
&= \sum_{\act' \in \Actions((\hist\act))} \sum_{\term \in \Terminals[(\hist\act\act')]} \reach{\strategy^{\time+1}}((\hist\act), \term) \partutil{\time+1}{\term}\\
&= \sum_{\term \in \Terminals[(\hist\act)]} \reach{\strategy^{\time+1}}((\hist\act), \term) \partutil{\time+1}{\term}
\end{align*}
In the last step we use that $\Terminals[(\hist\act)]$ is partitioned into the sets $\Terminals[(\hist\act\act')]$ by which action $\act' \in \Actions((\hist\act))$ follows $(\hist\act)$. The lemma follows by induction.
\end{proof}

Next, we show that the predictive baseline update maintains an invariant that the baseline values are a weighted sum of values $\partutil{\time}{\term}$.

\begin{lem}
For any time step $\time$, the predictive baseline satisfies
\begin{equation}
\baseline^{\time}(\hist, \act) = \sum_{\term \in \Terminals[(\hist\act)]} \reach{\strategy^{\time}}((\hist\act), \term) \partutil{\time}{\term}
\end{equation}
\label{lem:predinvariant}
\end{lem}

\begin{proof}
We proceed by induction on time. For the base case, $\time = 1$, by definition $\baseline^\time(\hist, \act) = 0$, and $\partutil{\time}{\term} = 0$ for all $\term \in \Terminals$.

For the inductive step, we assume that the lemma holds at time $\time$, and we show that it then follows for time $\time + 1$. We break this into two cases, based on whether $(\hist\act)$ is sampled on time $\time$ or not.

If $\Pubstate((\hist\act)) \not\sqsubseteq \Terminals^{\time}$, then $\baseline^{\time+1}(\hist, \act) = \baseline^{\time}(\hist, \act)$ by definition of the predictive baseline. Also by definition $\partutil{\time+1}{\term} = \partutil{\time}{\term}$ for any $\term \in \Terminals[(\hist\act)]$, because otherwise $\term \in \Terminals^\time$. Next, we show that $\reach{\strategy^{\time+1}}((\hist\act), \term) = \reach{\strategy^{\time}}((\hist\act), \term)$ for any $\term \in \Terminals[(\hist\act)]$. Assume for the sake of contradiction that this doesn't hold. Then there must be some $(\hist'\act') \sqsupseteq (\hist\act)$ such that $\strategy^{\time+1}(\hist',\act') \neq \strategy^{\time}(\hist',\act')$. By the definition of the MCCFR algorithm, this only occurs if $\Regret{\time+1}(\Iset(\hist'),\act') \neq \Regret{\time}(\Iset(\hist'),\act')$, which in turn only occurs if $\Pubstate(\hist') \sqsubseteq \Terminals^\time$. But then $\Pubstate((\hist\act)) \sqsubseteq \Pubstate(\hist') \sqsubseteq \Terminals^\time$, which contradicts the premise. Putting this all together, we have
\begin{align*}
\baseline^{\time+1}(\hist, \act) &= \baseline^{\time}(\hist, \act)\\
&= \sum_{\term \in \Terminals[(\hist\act)]} \reach{\strategy^{\time}}((\hist\act), \term) \partutil{\time}{\term} &&\text{by inductive hypothesis}\\
&= \sum_{\term \in \Terminals[(\hist\act)]} \reach{\strategy^{\time+1}}((\hist\act), \term) \partutil{\time+1}{\term}
\end{align*}
This completes the inductive step in the case that $(\hist\act)$ was not sampled.

If $\Pubstate((\hist\act)) \sqsubseteq \Terminals^{\time}$, then the following holds
\begin{align*}
\baseline^{\time+1}(\hist, \act) &= \baseutil{(\hist\act)}{\strategy^{\time+1}}{\Terminals^\time} &&\text{by predictive baseline definition}\\
&= \sum_{\term \in \Terminals[(\hist\act)]} \reach{\strategy^{\time+1}}((\hist\act), \term) \partutil{\time+1}{\term} &&\text{by Lemma~\ref{lem:predupdate}}
\end{align*}
This completes the inductive step in the case that $(\hist\act)$ was sampled. Thus the inductive step always holds, and the lemma follows by induction.
\end{proof}

To prove Theorem~\ref{thm:zerovar}, we note that if $\Terminals[\hist] \subseteq \bigcup_{\tau < \time} \Terminals^\tau$, then by definition $\partutil{\time}{\term} = \utility(\term)$ for any $\term \in \Terminals[\hist]$. Thus we have
\begin{align*}
\baseline^\time(\hist,\act) &= \sum_{\term \in \Terminals[(\hist\act)]} \reach{\strategy^{\time}}((\hist\act), \term) \partutil{\time}{\term} &&\text{by Lemma~\ref{lem:predinvariant}}\\
&= \sum_{\term \in \Terminals[(\hist\act)]} \reach{\strategy^{\time}}((\hist\act), \term) \utility(\term)\\
&= \exputil{(\hist\act)}{\strategy^\time} &&\text{by definition}
\end{align*}\qed

\section{Measuing counterfactual value variance}
\label{sec:empvar}

In Figure~2b of the main paper, we reported empirical variance of counterfactual values for POS MCCFR with baselines. To measure these, we run MCCFR for some number of iterations, then freeze the strategy. We then walk every information set-action pair in the game tree, and for each such pair we run a large number of sampled trajectories originating at the pair. These trajectories are walked as if we were running MCCFR with POS, but we do not update the strategy. Instead, we only calculate the sampled counterfactual value $\sum_{\hist \in \Iset} \reach{\strategy^\time}_{-\Toact(\hist)}(\hist)\baseutil{(\hist\act)}{\strategy^\time}{\term^\time}$ at the initial $\Iset, \act$ pair. From these samples, we compute an estimate of the true variance of the counterfactual value. Finally, we average these variance estimates across all information set-action pairs in the game.

\section{Baselines in Monte Carlo continual resolving}

Typically, online play in games with perfect information is accomplished by performing an independent computation at each new decision point to decide the agent's next action. Traditionally, this approach was intractable in games with imperfect information because there was no way to guarantee that these individual decisions would fit together into a cohesive equilibrium strategy. Instead, the traditional way of playing online was to use a precomputed strategy as a lookup table. Recently, however, techniques have been developed for safe and efficient online computation of strategies in imperfect information games \citep{Moravcik17, Brown17, Brown18b}. In this section, we discuss how Monte Carlo baselines fit into this work.

A particular example of this new paradigm, \emph{continual resolving}, was used in the DeepStack agent which defeated poker professionals \citep{Moravcik17}. Continual resolving contains two key parts. First, a safe resolving method is used to compute each decision independently in an online fashion, while still guaranteeing that the agent plays an approximate equilibrium strategy. Second, value approximation is used to restrict the depth of future actions that are considered when solving each decision. With these ingredients, CFR+ is used to solve a relatively small subgame each time the agent must make an action selection. \citet{Sustr19} replaced the CFR+ solver with MCCFR, creating \emph{Monte Carlo continual resolving (MCCR)}. It is straightforward to use our baseline framework within MCCR.

We conducted an experiment examining MCCR with baselines. We performed our experiment in Leduc. We measure the exploitability of strategy profiles that are constructed by independently solving each decision point as in online self-play. For each decision, we solve a subgame of depth three (i.e. looking a maximum of three actions into the future). After three actions, we approximately a history by running 100 iterations of CFR+ on the subtree rooted at the history and using the resulting strategy\footnote{This strategy, which contain errors because of the low number of iterations, approximates using a neural net for evaluation.} to generate values. At each decision point, we run resolving until we have performed a maximum number of evaluations, either at terminal histories or depth-limited evaluations. We use this as an implementation-independent way of comparing algorithms, as evaluations take the vast majority of computation time in continual resolving.

We compare MCCR with and without baselines. We use CFR+ updates, which we've found to decrease variance when combined with any baseline, and public outcome sampling with transitions sampled from the uniform distribution. Because of the inexact nature of the evaluation function, Theorem~\ref{thm:zerovar} does not hold in this setting, and we found learned history baselines to slightly outperform predictive baselines in preliminary experiments. We also compare to (deterministic) continual resolving, with both CFR and CFR+ update rules.

\begin{figure*}[htbp]
\centering
\includegraphics[width=\linewidth]{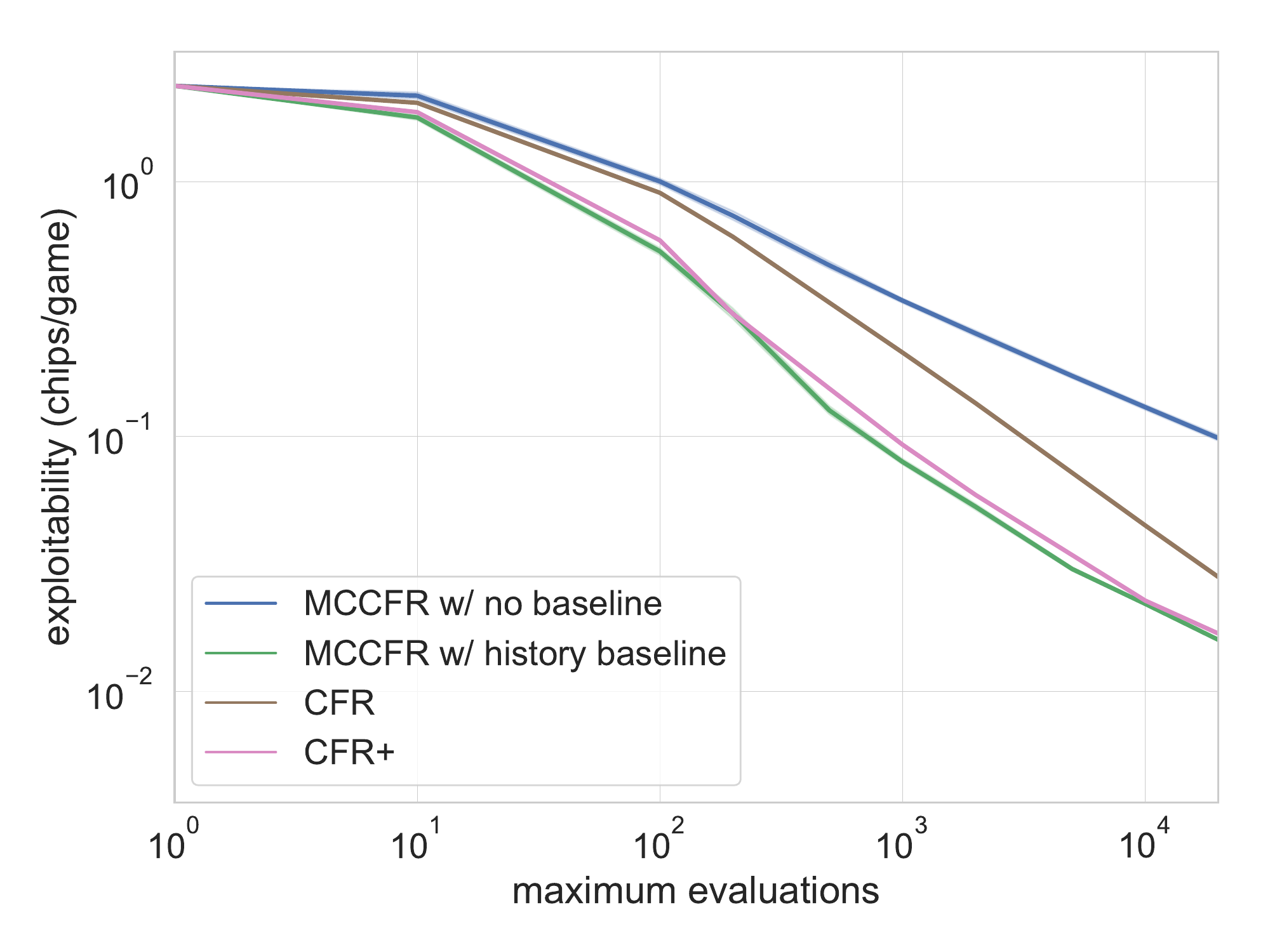}
\caption{Exploitability of continual resolving strategies based on the maximum number of evaluations allowed per resolve. For MCCFR, results are averaged over 20 runs with 95\% confidence intervals shown as (indiscernible) error bands.}
\label{fig:resolving}
\end{figure*}

Results are shown in Figure~\ref{fig:resolving}. We see that the inclusion of a baseline significantly decreases the exploitability of MCCR strategies. Without a baseline, MCCR is not competitive with the deterministic continual resolving. With a baseline, it is able to clearly outperform continual resolving with CFR updates, and slightly outperform continual resolving with CFR+ updates. This is especially notable because there is still plenty of room to improve the technique, such as by tuning the baseline and especially by refining the sampling strategy.
}

\end{document}